\documentclass[12pt]{article}

\newcommand{\outputDir}{output}

\usepackage[margin=.9in]{geometry}
\usepackage{XCharter} % better font
\usepackage[english]{babel}
\usepackage{ae,aecompl,amsmath,amsbsy,amssymb,eurosym,amssymb,amsthm,bm}
\usepackage{mathtools}

\usepackage{comment}
\usepackage{color}
\usepackage[dvipsnames]{xcolor}
\usepackage{pdflscape}
\usepackage{enumerate}
\usepackage{rotating}
\usepackage{adjustbox,algorithm2e}

\usepackage{setspace}
\usepackage{url}
\usepackage{xr} % for external cross-references
\usepackage{xr-hyper}
\usepackage[colorlinks]{hyperref}
\hypersetup{
  colorlinks=true,
  linkcolor=blue,
  citecolor=blue,      
  urlcolor=cyan,
  filecolor=orange % external references
}
\usepackage{tikz}
\usetikzlibrary{positioning,calc}

\usepackage{fancyhdr}
\usepackage{bbm}
\usepackage{authblk}
\usepackage[multiple]{footmisc}
\usepackage[round]{natbib}
\usepackage[]{caption}
\usepackage{graphicx}
\usepackage{longtable}
\usepackage{float}
\usepackage{tikz}
\usepackage{multirow}
\usepackage[flushleft]{threeparttable}
\usepackage{subcaption}
\usepackage{booktabs}
\usepackage{tabularx}
\newcolumntype{s}{>{\hsize=.33\hsize}X}
\usepackage{tabularray}
\usepackage{float}
\usepackage{graphicx}
\usepackage{codehigh}
\usepackage[normalem]{ulem}
\UseTblrLibrary{booktabs}
\UseTblrLibrary{siunitx}

\NewTableCommand{\tinytableDefineColor}[3]{\definecolor{
#1}{#2}{#3}}
\usepackage{setspace}
\usepackage{tikz}
\usetikzlibrary{matrix, positioning,patterns}
\makeatletter
\newcommand*{\addFileDependency}[1]{% argument=file name and extension
  \typeout{(#1)}
  \@addtofilelist{#1}
  \IfFileExists{#1}{}{\typeout{No file #1.}}
}
\makeatother

\makeatletter
\@ifclassloaded{beamer}{
  
}
\makeatother

\newtheorem{assumption}{Assumption}[section]

\newtheorem{lemma}{Lemma}
\newtheorem{prop}{Proposition}
\newtheorem{corollary}{Corollary}

\theoremstyle{definition}

\newtheorem{remark}{Remark}[section]

\newcommand{\E}{\mathbb{E}}
\newcommand{\Var}{\mathbb{V}}
\newcommand{\Cov}{\mathbb{C}\text{ov}}

\newcommand{\argmin}{\mathop{\mathrm{argmin}}}

\newcommand{\plim}{\text{plim}}
\newcommand{\norm}[1]{\left\|#1\right\|}
\newcommand{\normtwo}[1]{\norm{#1}_2}
\newcommand{\normF}[1]{\norm{#1}_F}

\DeclareMathOperator{\vectorize}{vec} % Vectorization
\DeclareMathOperator{\im}{Im} % Image space

\begin{document}
\onehalfspacing

\title{\textbf{On Causal Inference with Model-Based Outcomes}\thanks{{\small We thank Manuel Arellano, Stephane Bonhomme, Kirill Borusyak, Nezih Guner, Lance Lochner, Patrick Kline, Victor Sancibrian, Pedro Sant'Anna, Davide Viviano, and Dean Yang for their feedback, which has greatly improved the paper. We also thank seminar participants at Harvard, UCLA, Carlos III, CEMFI, IFAU, Tel-Aviv U., Hebrew U., and Warwick, and participants in the 15th Research Workshop Banco de Espa\~na - CEMFI, ELSE 2024, SAE 2024, Econometric Society Winter meeting 2024, Barcelona GSE Summer Forum 2024, ASSA 2025, Causal Inference with Panel Data workshop 2025 at Emory, and JSM 2025 for their helpful comments. This paper incorporates and extends the ideas developed in the previously circulated draft ``Using Event Studies as an Outcome in Causal Analysis'' by the same authors. We gratefully acknowledge financial support from the Maria de Maeztu Unit of Excellence CEMFI MDM-2016-0684 Grants CEX2020-001104-M, PID2021-128992NB-I00, PID2022-143184NA-I00, and RYC2022-037963-I; Funded by MICIU/AEI/10.13039/501100011033, ESF+,  and by “ERDF/EU”, which covered data access to the Dutch Central Bureau of Statistics and travel expenses to present the project. 
All views and any errors are our own.}} }

\author{Dmitry Arkhangelsky\thanks{{\small  CEMFI, \url{darkhangel@cemfi.es}, corresponding author.}} 
\vspace{-1.0em}
\and Kazuharu Yanagimoto\thanks{{\small CEMFI, \url{kazuharu.yanagimoto@cemfi.edu.es}.}}
\vspace{-1.0em}
\and Tom Zohar \thanks{{\small CEMFI, CESifo, \url{tom.zohar@cemfi.es}.}}}

\vspace{-2em}
\date{
\today }
\maketitle\thispagestyle{empty}

\vspace{-2em}

\begin{abstract}
\footnotesize
	We study the estimation of causal effects on group-level parameters identified from microdata (e.g., child penalties). We demonstrate that standard one-step methods (such as pooled OLS and IV regressions) are generally inconsistent due to an endogenous weighting bias, where the policy affects the implicit weights (e.g., altering fertility rates). In contrast, we advocate for a two-step Minimum Distance (MD) framework that explicitly separates parameter identification from policy evaluation. This approach eliminates the endogenous weighting bias and requires explicitly confronting sample selection when groups are small, thereby improving transparency. We show that the MD estimator performs well when parameters can be estimated for most groups, and propose a robust alternative that uses auxiliary information in settings with limited data. To illustrate the importance of this methodological choice, we evaluate the effect of the 2005 Dutch childcare reform on child penalties and find that the conventional one-step approach yields estimates that are substantially larger than those from the two-step method.
\end{abstract}

\vspace{1em}
\vspace{0.5em} 
\noindent
\clearpage

\section{Introduction}
\label{sec:intro}
A significant body of empirical work assesses policies implemented at a group level (e.g., firm, school, local labor market, or municipality) and their impact on aggregate group-level outcomes. Rather than simple averages of individual data, research interest often centers on specific parameters estimated from micro-level models; examples include firm-specific wage premia \citep{Abowd1999}, bank-specific credit supply \citep{khwaja2008tracing}, firm-level productivity \citep{amiti2007trade}, or the ``child penalty'' reflecting labor market consequences following childbirth \citep{Kleven2019a}. These model-based parameters represent key group-level features that may be shaped by the policy of interest. Common estimation strategies involve either one-step approaches, such as micro-level regressions with policy interactions, or two-step methods, where group-specific parameters are estimated first before being related to the policy in a second stage.

Estimating causal effects on such model-based outcomes faces a core complication: the policy typically shifts the entire distribution of the underlying micro-level data. For instance, labor policies can alter firm-level employment rates relevant for wage premia, while childcare expansion can change the average fertility rate in a given municipality. This distributional response fundamentally challenges causal analysis, raising key questions: How to define causal effects when the underlying population distribution is policy-dependent? And under what conditions do standard one-step or two-step estimation strategies identify such effects?

This paper addresses these questions by providing a systematic approach to causal inference for policy effects on model-based group-level parameters. We make three main contributions.

First, we develop a formal econometric framework that explicitly defines group-specific parameters through linear population moment conditions (i.e., a model). We posit that each counterfactual policy level yields a distinct potential distribution of micro-level data. Since the parameters of interest are functionals of these distributions, this mapping induces a causal structure on the group-level objects themselves (i.e., model-based outcomes). To define the causal effect of interest, we introduce an "oracle" estimand: the effect that would be recovered by a given policy evaluation strategy if the group-specific parameters were directly observed rather than estimated. This precise definition allows us to separate the structural impact of the policy on the parameters of interest from its influence on the implicit method-specific weights used in one-step estimation.

Second, we demonstrate a fundamental problem with standard one-step estimation methods, such as OLS with policy interactions. By analyzing these approaches within the Generalized Method of Moments (GMM) framework, we uncover a novel source of inconsistency arising from endogenous weighting: by altering the aggregate features of the micro-level data, the policy generically affects the implicit GMM weights. This correlation between the policy and the effective weights renders the estimator inconsistent for the target causal effect. This mechanism is analogous to a "bad control" problem \citep{Angrist2009}, but one that operates through the estimator's weighting matrix rather than the covariates.

Third, we analyze the properties of two-step Minimum Distance (MD) estimators. While the relationship between one-step and two-step estimators in linear regression models has a long history in econometrics (e.g., \citealp{amemiya1978note,Arellano2012,Chamberlain1992,de1986random,Graham2012,kline2024firm,Muris2022}), this literature has focused on statistical challenges. Instead, we re-evaluate this choice through the lens of causal analysis, focusing on a broader class of models with linear moment conditions. This change in perspective produces a new mechanism---policy-induced endogenous weighting---that is mechanically absent in prior work. We argue that the two-step estimator is strictly preferable for policy evaluation because it decouples parameter identification from the estimation of policy effects, allowing for transparent, exogenous weighting that eliminates the bias of the one-step estimator.

Crucially, this robustness requires explicitly confronting data limitations that one-step methods obscure. When group sizes are small, group-specific parameters cannot be estimated for some of the groups. This feature is not merely a theoretical concern but a practical reality, even with administrative datasets comprising millions of observations, as researchers focus on narrowly defined subpopulations. While one-step procedures hide these data failures— through opaque weighting or imputation—the two-step estimator makes the selection problem transparent. We demonstrate that when the share of discarded groups is small (i.e., most groups are sufficiently large), this concern vanishes alongside the within-group estimation noise. In this regime, the two-step estimator is asymptotically equivalent to the infeasible "oracle" that observes the true parameters; this allows researchers to proceed as if the outcomes were known, thereby permitting the use of standard inference techniques (e.g., cluster-robust standard errors). Conversely, when most groups are small, the explicit nature of the two-step approach allows researchers to diagnose the selection problem and, if necessary, stabilize the estimator using auxiliary population information.

To demonstrate the practical relevance of our theoretical findings, we present an empirical evaluation of the 2005 Dutch childcare expansion on "child penalty" measures. This policy setting provides a suitable case study in which the reform likely influenced both the outcome of interest (parental labor supply) and the distribution of the underlying individual-level data (fertility decisions). We find that a conventional one-step estimation strategy (similar to those in \cite{Kleven2024,Rabate2021}) yields estimates of the policy's impact that are substantially larger than those obtained using our robust two-step approach. While we do not claim this divergence is driven exclusively by the bias, the economic magnitude of the discrepancy highlights that results in such settings are highly sensitive to the estimator's implicit weighting scheme. Given this sensitivity, we argue that the transparency and flexibility of the two-step estimator make it the safer choice for applied policy evaluation.

Our analysis of one-step estimators relates to the recent literature on DiD methods dealing with heterogeneous treatment effects (e.g., \citealp{Borusyak2024}, \citealp{Callaway2021}, \citealp{dechaisemartin2020}, \citealp{goodman2021difference}, \citealp{Sun2021}; see \cite{Arkhangelsky2024b} for a survey) and related literature on heterogeneity in linear regressions \citep{Goldsmith-Pinkham2024}. While the mechanism we identify—policy-dependent weighting—is conceptually related to the concerns raised in \cite{Goldsmith-Pinkham2024}, our setting presents a distinct challenge. The standard remedy of saturating the model with interactions is structurally infeasible here because the group-level policy is perfectly collinear with group fixed effects. Consequently, the bias we identify necessitates a different solution, thereby motivating our focus on the two-step approach.

Our findings also connect to the literature on GMM estimation challenges, particularly concerning weighting matrices (e.g., \citealp{altonji1996small,newey2004higher}). Unlike work focusing on finite-sample biases that arise from using estimated optimal weights, the GMM inconsistency we document can occur even with simple, pre-specified weights. It is driven by the policy affecting the population moments within the GMM objective function. Similarly, the related asymptotic inconsistency we find for two-step estimators in small groups stems from policy-induced sample selection affecting the first stage, rather than challenges specific to estimating optimal second-stage weights.

Finally, our proposal of using auxiliary population moments to stabilize the two-step estimator connects to design-based identification strategies. It parallels design-based estimation strategies, where known assignment probabilities are used to address composition issues, and resonates with recent approaches that incorporate known population quantities directly into estimation (e.g., \citealp{Arkhangelsky2024c,Borusyak2023,Borusyak2024a}). More broadly, as demonstrated by \cite{imbens1994combining}, aggregate information generically improves micro-based estimation; in our model, the availability of such information mitigates inconsistency arising from finite group sizes.

The remainder of the paper is organized as follows: we begin by developing our framework and illustrate it with a motivating example (Section~\ref{sec:framework}). We then analyze the properties of one-step GMM estimators (Section~\ref{sec:identification}), and two-step Minimum Distance estimators (Section~\ref{sec:md_estimator}), highlighting pitfalls related to endogenous weighting and sample selection. Section~\ref{sec:empirics} provides an empirical illustration using Dutch administrative data. Finally, Section~\ref{sec:conclusion} concludes. All the technical proofs are collected in Appendix~\ref{app:theory}.

\section{Framework}
\label{sec:framework}
In this section, we develop an econometric framework for estimating causal effects on model-based group-level parameters. We proceed in two steps. First, we present a motivating example that guides our analysis.\footnote{To streamline the presentation, in the main text we focus on a single example, with Appendix~\ref{app:examples} discussing additional examples relevant for applications in labor, banking, and international trade.} Second, we formalize the data-generating process and define the target parameters. This allows us to rigorously define the ``oracle'' estimator that researchers aim to mimic and to contrast the properties of two feasible estimation strategies: a two-step Minimum Distance (MD) estimator and a one-step Generalized Method of Moments (GMM) estimator.

A central feature of our analysis is the hierarchical structure of the data: the policy variation of interest operates at the aggregate level $g$, yet identification relies on unit-level data indexed by $(g,i)$. We adopt this explicit double-index notation to highlight two features of the economic environment that a usual single-index notation obscures. First, identification requires us to distinguish the assumptions regarding policy assignment (the ``aggregate'' variation, $g$) from those governing individual choices (the ``micro'' variation, $i$). Second, the explicit group index allows us to confront the statistical reality of group-level dependence and varying sample sizes. Data within groups are naturally dependent because units share common policy shocks. Furthermore, the magnitude of the group sizes dictates the relevant statistical regime: when groups are large, within-group estimation uncertainty is negligible, but when groups are small, this uncertainty is first-order, and parameters may not be identified at all. Our analysis directly addresses these challenges.

\subsection{Motivating Example: Child Penalties}\label{subsec:motivating_example}

Consider the estimation of the ``child penalty'' -- the effect of motherhood on women's labor market outcomes (e.g., \citealp{Kleven2019a}). Let $Y_{g,i,t}$ be the outcome (e.g., earnings) for individual $i$ in municipality $g$ at time $t \in \{1, 2\}$. Let $E_{g,i}$ be an indicator equal to one if the individual experiences a first birth between periods 1 and 2. A standard DiD specification relates outcomes to individual fixed effects $\gamma_{g,i}$, group-time effects $\delta_{g,t}$, and the event:
\begin{align*}
    Y_{g,i,t} = \gamma_{g,i} + \delta_{g,t} + \tau_g E_{g,i} \mathbf{1}\{t=2\} + \epsilon_{g,i,t}.
\end{align*}
The parameter of interest, $\tau_g$, represents the average child penalty in municipality $g$. Identification of $\tau_g$ relies on the regression in differences: 
\begin{equation*}
    \Delta Y_{g,i} = \Delta \delta_g + \tau_g E_{g,i} + \Delta \epsilon_{g,i}.
\end{equation*}

Suppose we wish to evaluate the effect of a randomly assigned municipality-level policy $W_g$ (e.g., childcare provision) on the magnitude of the child penalty. Consider a thought experiment where the true group-specific penalties $\tau_g$ were directly observable. In this idealized setting, the policy analysis would be trivial: one would simply regress the penalty on the policy:
\begin{align*}
    \tau_g = \alpha + \beta W_g + \nu_g.
\end{align*}
With $W_g$ being randomly assigned, $\beta$ captures an average causal effect of childcare availability on the career costs of parenthood. The problem with this analysis is, of course, that $\tau_g$ is not directly observed in the data.

Lacking observations of $\tau_g$, a common empirical strategy is to substitute the model $\tau_g = \alpha + \beta W_g$ directly into the micro-level equation. This leads to a one-step estimating equation where the researcher pools data from all municipalities:
\begin{align*}
 \Delta Y_{g,i}  = \Delta\delta_g +(\alpha + \beta W_g) E_{g,i} + \varepsilon_{g,i}.
\end{align*}
By estimating this equation via OLS, this strategy implicitly weights each municipality by the variance of the treatment indicator $E_{g,i}$ (i.e., the fertility rate). If the policy $W_g$ (childcare) encourages fertility, it systematically shifts these weights toward treated municipalities, effectively destroying the original random assignment of $W_{g}$ and yielding an inconsistent estimator.

This weighting concern is not specific to OLS estimation; it extends naturally to Instrumental Variable (IV) settings. In many contexts, fertility decisions $E_{g,i}$ are viewed as endogenous to labor market potential (e.g., \citealp{lundborg2017can}). To address this, a researcher might instrument fertility using a shock $Z_{g,i}$ (e.g., IVF success). The ``One-Step'' strategy for estimating the policy effect $\beta$ remains conceptually identical: the researcher estimates the pooled equation above using Two-Stage Least Squares (TSLS). Here, the endogenous interaction term is $(\alpha + \beta W_g) E_{g,i}$, and the researcher employs $Z_{g,i}$ and the interaction $Z_{g,i} W_g$ as instruments. 

However, just as the OLS estimator is weighted by the variance of fertility, the IV estimator is implicitly weighted by the covariance between the instrument and the treatment (the strength of the first stage). If the policy $W_g$ affects the compliance rate---for instance, by making additional IVF rounds more affordable---the effective weights shift endogenously, yielding an inconsistent estimator even if $W_g$ is randomly assigned.

An alternative strategy, which we formalize as Minimum Distance (MD), aligns more closely with the thought experiment described above and accommodates both OLS and IV identification. One first estimates $\hat{\tau}_g$ within each municipality (using either OLS or IV as appropriate) and then regresses these estimates on $W_g$. As we show below, this separates parameter identification from policy evaluation, offering robustness to the composition and compliance effects that one-step approaches lack.

\begin{remark}[The role of heterogeneity] \label{rem:did_heterogeneity}
For simplicity, the exposition above assumes a constant effect $\tau_g$, thereby demonstrating that the biases we describe are not attributable to unit-level heterogeneity. If true effects vary across units, the group-specific moments capture an average effect (ATE in OLS, or LATE in IV). A policy $W_g$ might then influence this average not only through changes in individual effects $\tau_{g,i}$ but also by shifting the population composition relevant for the average. Appendix~\ref{sec:aggregate_interpretation} discusses these direct and indirect channels.
\end{remark}

\subsection{Formal Setup}\label{subsec:formal_setup}

We now describe a general framework behind the motivating example. Consider a collection $\mathcal{G}$ of $G$ groups, indexed by $g$. Within each group $g$, we observe data $D_{g,i} \in \mathcal{D}$ for $i=1, \dots, n_g$ units, drawn i.i.d. from a group-specific distribution $F_g$. A group-level policy $W_g \in \mathcal{W} \subseteq\mathbb{R}^p$ is assumed to causally affect this data-generating process via potential outcome distributions $F_g(w)$ for $w \in \mathcal{W}$. The observed distribution satisfies 
\begin{align*}
    F_g = F_g(W_g),
\end{align*}
effectively imposing a group-level Stable Unit Treatment Value Assumption (SUTVA).

We focus on a $k$-dimensional parameter vector $\bm{\theta}_g\in \mathbb{R}^k$. We define this parameter implicitly via $k$ population moment conditions:
\begin{equation*} 
\E_{F_g}[h(D_{g,i}, \bm{\theta}_g)] = \mathbf{0}_k,
\end{equation*}
where $h: \mathcal{D} \times \Theta \to \mathbb{R}^k$ is a known function. This definition anchors $\bm{\theta}_g$ directly to the underlying data distribution.\footnote{Throughout the text, we index within-group expectations using $F_g$. Formally, any such expectation is a random variable due to its dependence on $W_g$, which induces a random distribution $F_g(W_g)$. We use expectations without the $F_g$ subscript to average over group-level uncertainty, which always includes $W_g$.}
Since $F_g$ depends causally on $W_g$, this naturally generates potential outcomes for the parameter:
\begin{align*}
\bm{\theta}_g(w)\,\text{ solves }\,\E_{F_g(w)}[h(D_{g,i}, \bm{\theta}_g(w))] = \mathbf{0}_k.
\end{align*}
Defining the outcome $\bm{\theta}_g(w)$ via moment conditions clarifies its economic interpretation, which has important implications for downstream causal analysis (see Appendix~\ref{sec:aggregate_interpretation}). We focus on moment functions linear in the parameter:
\begin{align*}
h(D_{g,i}, \bm{\theta}) = h_{1}(D_{g,i}) - h_{2}(D_{g,i})\bm{\theta},
\end{align*}
which allows us to include both OLS and IV-type strategies discussed in the previous section. 
Letting $H_{j,g}(w) := \E_{F_{g}(w)}[ h_{j}(D_{g,i})]$ and assuming $H_{2,g}(w)$ is invertible, the potential outcome for $\bm{\theta}_g$ is
\begin{align*}
    \bm{\theta}_g(w) = H_{2,g}(w)^{-1} H_{1,g}(w).
\end{align*}
The interpretation of $\bm{\theta}_g(w)$ is, of course, application-specific; in our empirical example, it corresponds to a child penalty, identified either through OLS or IV restrictions on earnings and fertility decisions.\footnote{While we will assume that $H_{2,g}(w)$ is invertible throughout the paper, the same might not be true for its sample counterpart---an issue we will discuss at length in Section~\ref{sec:md_estimator}.}

To anchor our analysis, we consider an idealized benchmark estimator representing what applied researchers might compute if the true group-specific parameters $\bm{\theta}_g = \bm{\theta}_g(W_g)$ were directly observable. This benchmark, which we term the ``oracle'' estimator, takes the form of linear regression motivated by common practice in policy evaluation:
\begin{equation}\label{eq:main_object}
(B^{\star}, \bm{\alpha}^{\star}, \{\bm{\lambda}^{\star}_g\}) := \argmin_{\substack{B, \bm{\alpha}, \{\bm{\lambda}_g\} \\ \text{s.t. } \Gamma^\top \bm{\alpha} = 0, B \in \mathcal{B}_0}} \sum_{g=1}^{G}\left\|\bm{\theta}_g - (\bm{\alpha} + \Gamma \bm{\lambda}_g + B W_{g})\right\|_2^2.
\end{equation}
Here, $B$ ($k \times p$) captures the policy effects of primary interest which are a priori constrained to lie in a fixed linear subspace $\mathcal{B}_0$, $\bm{\alpha}$ ($k \times 1$) is an intercept term, $\Gamma$ ($k \times q$) captures the user-specified dimensions of unobserved group heterogeneity, and $\bm{\lambda}_g$ ($q \times 1$) represents group-specific coefficients (fixed effects) along these dimensions. The constraint $\Gamma^\top \bm{\alpha} = 0$ is a normalization imposed to guarantee a unique solution and does not affect the value of $B^{\star}$.  To ensure unique identification of $B^{\star}$, the parameter subspace $\mathcal{B}_0$ and the matrix $\Gamma$ must satisfy a structural alignment condition (Assumption~\ref{as:identification_subspace}). This requirement is automatically satisfied by standard specifications used in applied work.

Specification~\eqref{eq:main_object} encompasses standard approaches used for policy evaluation. For instance, if $\bm{\theta}_g$ is one-dimensional and the policy is randomly assigned, setting $\Gamma$ to zero yields the oracle specification discussed in the motivating example. Alternatively, if the analysis relies on a version of the parallel trends assumption rather than random assignment (which will be the case in our empirical analysis in Section~\ref{sec:empirics}), then $\Gamma\ne 0$ allows for group-level fixed effects.  In particular, if components of $\bm{\theta}_{g}$ correspond to different time periods, then \eqref{eq:main_object} includes a standard two-way equation,
\begin{equation*}
    \theta_{g,t} = \alpha_t + \lambda_g + \beta W_{g,t} + \nu_{g,t},
\end{equation*}
with $\Gamma = (1,\dots,1)^{\top}$ and $\mathcal{B}_0$ being a linear subspace spanned by the identity matrix (enforcing a constant $\beta$ across periods). In this case, the constraint $\Gamma^\top \bm{\alpha} = \sum_{t} \alpha_t = 0$ is a normalization that separates the time and group fixed effects and does not affect the value of $\beta$. Our framework also naturally extends to IV rather than OLS policy regressions, with Appendix~\ref{sec:extensions} illustrating this option.

While the oracle regression \eqref{eq:main_object} is infeasible in practice, it defines the target quantity of interest for our analysis. That is, we assume that \eqref{eq:main_object} is the procedure the researcher intends to implement but cannot because $\bm{\theta}_g$ is unknown. The causal interpretation of $B^{\star}$ depends on the nature of $W_{g}$ and specification \eqref{eq:main_object}. For example, in the case of a randomly assigned binary policy, $B^{\star}$ corresponds to the average (across groups) treatment effect (ATE). In contrast, in analyses that rely on parallel trends, the estimand is the average treatment effect on the treated (ATT). This difference, while crucial for the economic interpretation of $B^{\star}$, is not essential in our analysis. Our central question is whether, and under what conditions, feasible estimation strategies that rely (explicitly or implicitly) on estimates $\hat{\bm{\theta}}_{g}$ derived from micro-level data can mimic the desired oracle procedure in terms of estimation and inference.

Motivated by existing empirical practice, we consider two estimation strategies: a two-step Minimum Distance (MD) and a one-step Generalized Method of Moments (GMM). The MD estimator first obtains group-specific estimates $\hat{\bm{\theta}}_g$, then performs the second-stage regression:
\begin{equation}\label{eq:two-step}
(\hat {B}^{MD},\hat{\bm{\alpha}}^{MD}, \{\hat{\bm{\lambda}}^{MD}_g\}) := \argmin_{\substack{B, \bm{\alpha}, \{\bm{\lambda}_g\} \\ \text{s.t. } \Gamma^\top \bm{\alpha} = 0, B \in \mathcal{B}_0}} \sum_{g=1}^{G}\left\|\hat{\bm{\theta}}_g - (\bm{\alpha} + \Gamma \bm{\lambda}_g + B W_{g})\right\|_2^2.
\end{equation}
This estimator implements the constructive logic of the oracle: one first recovers the structural parameters $\hat{\bm{\theta}}_g$ (e.g., the municipality-specific child penalties $\hat{\tau}_g$ or industry-time productivity measures $\hat{\tau}_{g,t}$) and then projects them onto the policy $W_g$. 

The GMM estimator instead directly imposes the structure $\bm{\theta}_g = \bm{\alpha} + \Gamma \bm{\lambda}_g + B W_{g}$ onto the sample moments $h_{g,n}(\bm{\theta}) := n_g^{-1}\sum_{i=1}^{n_g} h(D_{g,i}, \bm{\theta})$:
\begin{multline}\label{eq:gmm_est}
(\hat B^{GMM},\hat{\bm{\alpha}}^{GMM}, \{\hat{\bm{\lambda}}^{GMM}_g\}) := \\
\argmin_{\substack{B, \bm{\alpha}, \{\bm{\lambda}_g\} \\ \text{s.t. } \Gamma^\top \bm{\alpha} = 0, B \in \mathcal{B}_0}} \sum_{g=1}^{G} h_{g,n}(\bm{\alpha} + \Gamma \bm{\lambda}_g + B W_{g})^\top A_g h_{g,n}(\bm{\alpha} + \Gamma \bm{\lambda}_g +B W_{g}),
\end{multline}
for specified weighting matrices $A_g$. This formulation mathematically formalizes the ``one-step'' strategies. For instance, in the child penalty example, minimizing \eqref{eq:gmm_est} is equivalent to running the single pooled OLS or TSLS regression of $\Delta Y_{g,i}$ on the interacted policy terms. The subsequent sections analyze the statistical properties of $\hat{B}^{MD}$ and $\hat{B}^{GMM}$, focusing on the conditions for consistency and inferential properties.
\begin{remark}[Choice between GMM and MD]
Applied researchers often conceptualize their analysis using the two-step logic of the MD estimator. However, implementation frequently relies on one-step GMM procedures, largely due to perceived statistical and computational convenience. \cite{Pekkarinen2009} exemplify this tendency in their study of education reform and intergenerational mobility. While they formally define their parameter of interest using a two-step framework (their Equations 1-2), they estimate it using pooled OLS (their Equation 3), explicitly citing the complexities of the two-step analysis as the rationale. 
\end{remark}
\begin{remark}[Weighted oracle]
Oracle specification \eqref{eq:main_object} does not rely on any group-specific weights. In some applications, we may wish to incorporate weights (e.g., upweighting larger groups), and our analysis can be updated straightforwardly to accommodate such weights. In particular, in the empirical application in Section~\ref{sec:empirics}, we re-weight the groups by their size, $n_g$.
\end{remark}

\section{GMM}\label{sec:identification}
This section analyzes the consistency properties of the GMM estimator \eqref{eq:gmm_est} as the number of groups $G \to \infty$. We consider an asymptotic regime in which the number of units per group is infinite ($n_g = \infty$). This simplifies the analysis by eliminating within-group estimation error, allowing us to treat the group-level population moments as directly observable, while maintaining the statistical uncertainty arising from the random assignment of policy across groups. We will incorporate the estimation error---which is arguably one of the main reasons why empirical researchers rely on one-step estimators---in Section~\ref{sec:md_estimator}.

\subsection{A Cautionary Tale for One-Step GMM}

With infinite $n_g$, the population moments $H_{1,g} = \E_{F_g}[h_1(D_{g,i})]$ and $H_{2,g} = \E_{F_g}[h_2(D_{g,i})]$ are directly observed, yielding the true group parameter $\bm{\theta}_g = H_{2,g}^{-1} H_{1,g}$. Therefore, the two-stage MD estimator \eqref{eq:two-step} using the true $\bm{\theta}_g$ coincides with the infeasible oracle estimator \eqref{eq:main_object}, ensuring $\hat{B}^{MD} = B^{\star}$.

However, the one-step GMM estimator introduces potential complications even when $n_g = \infty$. Substituting population moments into the GMM objective \eqref{eq:gmm_est}, the estimator minimizes a weighted sum of squared deviations from the target model:
\begin{align*}
     \min_{\substack{B, \bm{\alpha}, \{\bm{\lambda}_g\} \\ \text{s.t. } \Gamma^\top\bm{\alpha} = 0, B \in \mathcal{B}_0}}  \sum_{g=1}^{G}\left(\bm{\theta}_g - (\bm{\alpha} + \Gamma \bm{\lambda}_g + B W_{g})\right)^\top\tilde A_g \left(\bm{\theta}_g - (\bm{\alpha} + \Gamma \bm{\lambda}_g + B W_{g})\right),
\end{align*}
where the effective weighting matrix is $\tilde A_g := H_{2,g}^\top A_g H_{2,g}$.\footnote{Recall that $A_g$ is the original weighting matrix in \eqref{eq:gmm_est}, which is often defined implicitly by using OLS or TSLS.} This objective differs from the oracle target \eqref{eq:main_object}, which implicitly uses identity weights ($\tilde{A}_g = I_k$). Crucially, if the underlying moments $H_{2,g}$ depend on the policy $W_g$ (because $W_g$ changes the group-specific distribution $F_g$), then the effective weights $\tilde{A}_g$ may also depend on $W_g$, potentially leading to inconsistency.

To discuss the consistency of the GMM estimator, we adopt a simple data-generating process characterized by linear potential outcomes with common slopes:
\begin{assumption}\label{as:simple_model}
(a) Policy assignment $W_g$ is i.i.d. across groups $g$, independent of the potential outcome distributions $F_{g}(\cdot)$, with $\mathbb{V}[W_g]$ finite and positive definite.
(b) The potential distribution function $F_{g}(w)$ is i.i.d. across groups $g$ for any given $w$.
(c) The potential outcome function $\bm{\theta}_g(w)$ is linear in $w$ with a common slope $B_0$: $\bm{\theta}_{g}(w) = \bm{\alpha}_g + B_0 w$, where $B_0 \in \mathcal{B}_0$.
\end{assumption}
This assumption describes the ideal setup for policy evaluation: the policy is independent of the potential outcomes and is i.i.d. over groups. The potential outcome function is as simple as possible: linear in $w$ and with all variation coming from intercepts $\bm{\alpha}_g$, which can capture group-specific heterogeneity or the effect of other policy variables independent of $W_g$.

We now use this assumption to investigate when $\plim_{G\to\infty} \hat{B}^{GMM} = B_0$. Define the optimal population intercept $\bm{\alpha}^0$ and the associated residual $\bm{\epsilon}_g^0$ relative to the GMM weighting $\tilde{A}_g$:
\begin{align*}
\bm{\alpha}^0 &:= \argmin_{\bm{\alpha}: \Gamma^\top \bm{\alpha} = 0} \E\left[\min_{\bm{\lambda}_g}(\bm{\alpha}_g - \bm{\alpha} - \Gamma \bm{\lambda}_g)^\top \tilde A_g (\bm{\alpha}_g - \bm{\alpha} - \Gamma \bm{\lambda}_g)\right] \\
\bm{\epsilon}_g^0 &:= (\bm{\alpha}_g - \bm{\alpha}^0) - \Gamma \hat{\bm{\lambda}}_g^0(\bm{\alpha}^0), \
\end{align*}
where $\hat{\bm{\lambda}}_g^0(\bm{\alpha})$ minimizes the inner quadratic form for a given $\bm{\alpha}$. The expectation $\E[\cdot]$ is over the joint distribution of $(W_g, F_g(\cdot))$.
The residual $\bm{\epsilon}_g^0$ represents the component of baseline heterogeneity $\bm{\alpha}_g - \bm{\alpha}^0$ orthogonal to $\Gamma$ in the $\tilde{A}_g$-metric. GMM consistency hinges on the weighted correlation between this residual heterogeneity and the policy.
\begin{prop} \label{prop:consistency}
Suppose Assumption \ref{as:simple_model} holds. Then the probability limit $B_{lim} = \plim_{G\to\infty} \hat{B}^{GMM}$ satisfies $B_{lim} = B_0$ if and only if
\begin{equation} \label{eq:consistency_condition}
    M_{\Gamma}\Cov[\tilde A_g\bm{\epsilon}_g^0, W_g] = \mathbf{0}_k,
\end{equation}
where $M_{\Gamma} = I_k - \Gamma(\Gamma^\top\Gamma)^{-1}\Gamma^\top$ projects orthogonal to the columns of $\Gamma$, and the covariance is over the population distribution of groups.
\end{prop}
Proposition \ref{prop:consistency} formalizes the key issue: GMM is consistent if and only if the policy $W_g$ is uncorrelated with the heterogeneity component $\bm{\epsilon}_g^0$ in the GMM metric (i.e., weighted by $\tilde{A}_g$). Even under random assignment of $W_g$ (implying $\Cov[\bm{\alpha}_g, W_g] = 0$), this condition generally fails if the weights $\tilde{A}_g$ depend on $W_g$ and $\bm{\epsilon}_g^0 \ne 0$. The latter condition holds whenever $\bm{\alpha}_g \ne \bm{\alpha} + \Gamma \bm{\lambda}_g$, i.e., the model for $\bm{\theta}_g$ used in the GMM estimator is not exactly correct. Note that under Assumption~\ref{as:simple_model}, the probability limit of the oracle estimator, $\plim_{G\to \infty}B^{\star}$, is equal to $B_0$; thus, the MD estimator is consistent.

The discussion in this section illustrates a standard econometric result: in linear models, the MD and GMM estimators are algebraically interchangeable via appropriate weighting matrices. In particular, there exist matrices $\{A_g\}_{g=1}^{G}$  such that the two estimators are numerically the same. This fact, however, does not play a major role in empirical practice, which overwhelmingly relies on default choices.\footnote{Note that the standard optimal GMM weights are derived from the inverse moment variance and thus depend on the policy $W_g$, thereby enforcing the very biases we focus on.} Consequently, the distinction can be framed as a difference between an estimator where the relevant weighting matrix is explicitly specified by the user (MD), versus an estimator where the weighting matrix is specified by the algorithm (GMM). Proposition~\ref{prop:consistency} shows that the latter leads to bias because standard algorithms select a policy-dependent matrix. 

Beyond bias, the MD estimator has a basic advantage: the estimand is defined by the researcher rather than by the estimation algorithm. In practice, the correct causal model is likely nonlinear in $w$ and has additional dimensions of heterogeneity that Assumption~\ref{as:simple_model} abstracts from. These complications would affect both estimators (MD and GMM), and the analogue of Proposition~\ref{prop:consistency} for such a model would state that the two estimators converge to different limits. This fact is not disqualifying per se in a model with heterogeneous treatment effects or nonlinearity, because both estimators could, in principle, still converge to meaningful average effects. We focus on the simpler homogeneous model precisely because it reveals that the GMM estimator is inconsistent for the only possible effect of interest. Moreover, in a model with heterogeneous effects, the GMM estimator will exhibit additional biases, as we discuss in the next section.

\subsection{Motivating Example: Intuition for the Bias}
\label{sec:gmm_discussion}

Proposition~\ref{prop:consistency} establishes that GMM consistency requires the group-specific heterogeneity (captured by $\bm{\epsilon}_g^0$) to be uncorrelated with the policy $W_g$ in the $\tilde{A}_g$-weighted metric. To interpret this abstract condition, we return to the Child Penalty setting. In this context, the structural parameter $\tau_g$ represents the effect of motherhood ($E_{g,i}=1$) on earnings. We model the relationship between $\tau_g$ and the policy $W_g$ (e.g., childcare provision) as $\tau_g = \alpha_g + \beta_0 W_g$, where $\alpha_g$ captures unobserved group-level heterogeneity.

\paragraph{OLS: Variance Weighting}
When $\tau_g$ is identified via OLS (a standard approach in the literature), the effective weight of a municipality is determined by the variance of the treatment indicator:
\begin{equation*}
    \tilde A_g^{OLS} \propto \sigma_{E,g}^2(W_g) := \mathbb{P}_{F_g(W_g)}[E_{g,i}=1](1-\mathbb{P}_{F_g}[E_{g,i}=1]).
\end{equation*}

As a result, if the policy $W_g$ encourages fertility (e.g., cheaper childcare), it mechanically affects $\sigma_{E,g}^2(W_g)$. The potential average fertility, $\mathbb{P}_{F_g(w)}[E_{g,i}=1]$, is independent of $W_{g}$, because the latter is randomly assigned. However, the realized average fertility (and thus $\sigma_{E,g}^2$) is dependent on $W_{g}$ as long as the policy has any causal effect. The resulting bias is conceptually related to the ``bad control'' bias described by \cite{Angrist2009}: weighting by a post-treatment statistic ($\sigma_{E,g}^2$) is similar to conditioning on a post-treatment outcome in a regression. 

\paragraph{IV: Compliance Weighting}
The weighting mechanism extends to settings where fertility is viewed as endogenous at the individual level, and $\tau_g$ is identified via Instrumental Variables (e.g., using IVF success $Z_{g,i}$ as in \citealp{lundborg2017can}). In this case, the effective weight of a group is not determined by the variance of fertility, but by the strength of identification---specifically:
\begin{equation*}
    \tilde A_g^{IV} \propto \pi_g^2 \sigma^2_{Z,g}.
\end{equation*}
Here, the bias arises if the policy $W_g$ affects the compliance rate $\pi_g$.\footnote{The policy can also change $\sigma^2_{Z,g}$, but it is less likely in the case of the IVF treatment because this variance is determined primarily by biological constraints.} For example, this can happen if childcare subsidies make IVF more affordable for lower-income households. A one-step IV regression---where the interaction $E_{g,i} \times W_{g}$ is instrumented with $Z_{g,i} \times W_{g}$---will implicitly overweight groups with high compliance. As before, the resulting estimate $\hat{\beta}^{IV}$ will suffer from the induced ``bad control'' bias despite $W_g$ being randomly assigned. 

\paragraph{Statistical vs. Causal Weighting Bias}

Thus far, our discussion has assumed that the policy of interest (e.g., childcare provision) is randomly assigned. In practice, researchers often rely on weaker identifying assumptions, such as parallel trends or selection on observables. To illustrate the nuances introduced by these settings and to connect our findings to the existing literature on one- and two-step estimators, we consider a purposefully stylized model that describes the statistical relationship between a binary policy $W_g \in \{0,1\}$ and the baseline potential outcome $\alpha_g$: 
\begin{equation*}
    \alpha_g = \alpha + \epsilon_g,\quad \E[\epsilon_g|W_g] = 0.
\end{equation*}
This mean-independence restriction requires only that the average unobserved heterogeneity is balanced between treated and control groups, rather than being completely independent (as it would be if the policy were randomly assigned).

Focusing on the OLS estimator, unbiasedness requires that the effective weights---driven by the group-specific variance, $\sigma^2_{E,g}(W_g)$---do not introduce the correlation between $W_g$ and $\epsilon_g$. Specifically, the weighted errors must average to the same value across treatment arms:
\begin{equation*}
    \E[\sigma^2_{E,g}(1)\epsilon_g|W_g =1] =  \E[\sigma^2_{E,g}(0)\epsilon_g|W_g =0].
\end{equation*}
Adding and subtracting the counterfactual term $\E[\sigma^2_{E,g}(0)\epsilon_g|W_g =1]$ allows us to decompose the total bias into two distinct components:
\begin{equation*}
    \underbrace{\E[(\sigma^2_{E,g}(1)-\sigma^2_{E,g}(0))\epsilon_g|W_g =1]}_{\text{Endogenous Weighting Bias}} + \underbrace{\left(\E[\sigma^2_{E,g}(0)\epsilon_g|W_g =1] - \E[\sigma^2_{E,g}(0)\epsilon_g|W_g =0]\right)}_{\text{Statistical Weighting Bias}} = 0.
\end{equation*}
This decomposition highlights that the one-step estimator faces two distinct sources of bias, which operate through fundamentally different mechanisms.

The first term captures the endogenous weighting bias, the primary focus of this paper. This bias is causal in nature: it arises specifically because the policy alters the weights (i.e., $\sigma^2_{E,g}(1) \neq \sigma^2_{E,g}(0)$). Even if the policy is randomly assigned (ensuring the second term is zero), this causal channel remains active whenever the ``treatment effect'' on the weights correlates with baseline heterogeneity $\alpha_g$, as shown in Proposition~\ref{prop:consistency}.

The second term captures the statistical weighting bias. This bias is static in nature: it arises not from the policy's causal effect, but from a pre-existing correlation where the weighted errors are unbalanced across treatment arms. This component connects our framework to the literature on correlated random effects and linear regression models \citep{Arellano2012,Chamberlain1992,Graham2012,Muris2022}. Our results show that, unlike the statistical concerns emphasized in this literature, the causal channel remains a distinct problem even when statistical independence holds.

Comparing these cases yields a taxonomy of estimation biases. In the classical random effects analysis (e.g., \citealp{amemiya1978note,de1986random,hanushek1974efficient}), weights are neither statistically correlated with errors nor affected by the policy, eliminating both bias terms. In correlated random effects models, only the statistical bias is present. Our causal policy evaluation setting adds a new dimension to this discussion: under the random assignment, the statistical bias vanishes, but the endogenous bias arises solely due to the policy's impact on the weights. Finally, under the general mean-independence assumption common in applied work, both sources of bias are active.

\paragraph{Connection to the Contamination Bias}

Our framework also provides a complementary view of various ``heterogeneity biases'' recently highlighted in the literature (e.g., \citealp{dechaisemartin2020,Goldsmith-Pinkham2024, goodman2021difference}). In the specific case of OLS, our result connects to the ``contamination bias'' identified by \cite{Goldsmith-Pinkham2024} (GHK) in linear regressions with multiple treatments. In particular, expanding the one-step equation in the child penalty case yields:
\begin{align*}
\Delta Y_{g,i} = \delta_g +\alpha E_{g,i}+ \beta (W_{g}E_{g,i}) + \varepsilon_{g,i}.
\end{align*}
GHK demonstrate that in such regressions, the coefficient on the interaction term $W_g E_{g,i}$ absorbs the omitted heterogeneity from the main effect $\alpha E_{g,i}$. Proposition~\ref{prop:consistency} shows that, in our context, this effect is a specific instance of a broader GMM phenomenon: one-step estimators inevitably weight groups by their identification strength, which itself responds to the policy. Whether that strength comes from variance (OLS) or compliance (IV), the dependence of these weights on the policy $W_g$ generates bias.

Crucially, our hierarchical setting imposes distinct structural constraints that render existing remedies for the contamination bias infeasible. In a cross-sectional setting, GHK suggest removing contamination bias by saturating the model with interactions---effectively allowing the coefficient on $E_{g,i}$ to vary by $g$. In our hierarchical context, this would require including group-specific treatment effects, $\alpha_g E_{g,i}$, alongside the policy interaction, $W_g E_{g,i}$. Because the policy $W_g$ is constant within groups, the interaction term $W_g E_{g,i}$ is perfectly collinear with the vector of group-specific effects. Thus, one cannot ``control away'' the heterogeneity. Consistent estimation instead relies on the MD approach, which explicitly separates the identification of group-specific parameters from the evaluation of the policy.

Precisely this feature---the feasibility of constructive Minimum Distance estimation---shows that the hierarchical structure is a fundamental requisite for identification in our context, not merely a statistical feature of the sampling design. Consider the degenerate case without groups (or equivalently, with a single unit per group, $n_g=1$). In this setting, the first-stage estimation of $\bm{\theta}_g$ is impossible without introducing auxiliary information (a point we discuss in more detail in Section~\ref{subsec:design_info}). Without such information, the effect of the policy on the model-based parameter (e.g., the child penalty) is fundamentally unidentified. Thus, the groups are not just a clustering unit for inference; their presence is a structural condition that permits the separation of parameter identification from policy evaluation, thereby allowing us to address the endogenous weighting bias.

There is another interesting connection between our analysis and GHK. Suppose that in the context of Proposition~\ref{prop:consistency} we allow for heterogeneity in effects ($B_g$), not only intercepts ($\bm{\alpha}_g$). In this case, as long as $\Gamma \equiv 0$, it is straightforward to show that the MD estimator will continue to estimate an average (across groups) causal effect as long as $W_g$ is independent of the potential outcomes and is i.i.d. across groups. The GMM estimator, however, will exhibit an additional contamination bias because the endogenous weights will effectively render the distribution of $W_g$ group-specific, destroying the key property of the policy assignment mechanism.

\section{Minimum Distance Estimation}
\label{sec:md_estimator}
We now turn to the analysis of the Minimum Distance (MD) estimator \eqref{eq:two-step}. This two-stage procedure first obtains group-specific parameter estimates, $\hat{\bm{\theta}}_g$, and then uses them in a second-stage regression, mimicking the oracle objective \eqref{eq:main_object}. In the previous section, we assumed that the group size $n_g$ is infinite, which rendered the MD estimator identical to the oracle. In practice, however, $n_g$ is finite. Even in applications using comprehensive administrative data—such as our child penalty example—groups are often defined narrowly (e.g., by birth cohort, age at birth, gender, and education level, within a specific municipality), resulting in group sizes that can range from thousands to just a few units.

Acknowledging that groups are finite requires us to face a practical reality: for some groups, we may not be able to construct the estimate $\hat{\bm{\theta}}_g$ at all (e.g., a municipality with no births in a specific cell). Consequently, without additional information, these groups must be discarded from the second-stage estimation. While this event is rare for large groups, it can occur for a substantial share of the sample if many groups are small.

This creates a tension between two sources of error: the statistical noise inherent in the policy evaluation problem (which scales with the number of groups $G$) and the bias introduced by systematically discarding groups. In this section, we analyze the MD estimator under two statistical approximations that capture the extremes of this trade-off. First, we consider a regime where the share of discarded groups is large enough to overwhelm the statistical noise (Section~\ref{sec:selection}). This corresponds to an asymptotic analysis where $G \to \infty$ but $n_g$ remains bounded. We show that in this case, the selection bias dominates, rendering the MD estimator inconsistent.

Second, we consider a regime where the share of discarded groups is small relative to the statistical noise (Section~\ref{sec:md_inference}). This corresponds to an asymptotic scheme where $n_g$ increases with $G$. We show that in this case, the selection bias is dominated by the noise. Moreover, precisely because the groups are sufficiently large to avoid selection, the within-group estimation error becomes negligible, making the MD estimator statistically equivalent to the infeasible oracle. This result has immediate consequences for inference, simplifying the construction of standard errors. Finally, for settings where the share of discarded groups is not sufficiently small, we discuss how auxiliary information can be used to solve the identification problem (Section~\ref{subsec:design_info}). 

A key advantage of the MD estimator over the one-step GMM approach is that it forces the researcher to explicitly confront these data limitations by reporting summary statistics like the share of discarded groups, rather than burying them in a black-box weighting matrix.

\subsection{The Selection Problem}
\label{sec:selection}

Suppose the group-specific parameters $\bm{\theta}_g$ are identified via conditional moment restrictions:
\begin{equation*}
    \E_{F_g}[h_1(D_{g,i}) - h_{2}(\tilde D_{g,i})\bm{\theta}_g \mid \tilde D_{g,i}] = \bm{0}_k,
\end{equation*}
where $\tilde D_{g,i}$ is a subset of the data $D_{g,i}$. This structure is common in regression models and corresponds to the OLS version of the Child Penalty example in Section~\ref{subsec:motivating_example}. A natural first-stage estimator for $\bm{\theta}_g$ uses the sample analog within group $g$:
\begin{equation*}
    \hat {\bm{\theta}}_g := \left(\sum_{i=1}^{n_g} h_{2}(\tilde D_{g,i})\right)^{-1} \sum_{i=1}^{n_g} h_1(D_{g,i}) =: ( \hat{H}_{2,g})^{-1} (\hat{H}_{1,g}).
\end{equation*}
This estimator is only defined if the sample matrix $\hat{H}_{2,g}$ is invertible. Let 
\begin{align*}
    \omega_g := \mathbf{1}\left\{ \hat{H}_{2,g} \text{ is invertible}\right\}
\end{align*}
be the indicator for whether $\hat{\bm{\theta}}_g$ can be computed. In applied settings with small groups, the event $\omega_g=0$ represents a tangible data failure. For instance, in the Child Penalty example, $\hat{H}_{2,g}$ is singular if a relevant group (a subpopulation of a given municipality) has no recorded births in the sample.

When the estimator exists ($\omega_g=1$), the conditional moment restriction implies that $\hat{\bm{\theta}}_g$ is unbiased for the true parameter $\bm{\theta}_g$ conditional on the covariates:
\begin{equation*}
    \E_{F_g}[\hat {\bm{\theta}}_g \mid \omega_g = 1, \{\tilde D_{g,i}\}_{i=1}^{n_g}] = \bm{\theta}_g.
\end{equation*}
We can thus decompose the estimator into the true parameter and a conditionally mean-zero estimation error: $\hat {\bm{\theta}}_g = \bm{\theta}_g + \bm{\varepsilon}_{g}$. The feasible MD estimator \eqref{eq:two-step} must necessarily drop groups where $\hat{\bm{\theta}}_g$ is undefined. It therefore solves the weighted objective:
\begin{equation}\label{eq:feasib_md_weighted}
    (\hat {B}^{MD},\hat{\bm{\alpha}}^{MD}, \{\hat{\bm{\lambda}}^{MD}_g\})  :=
    \argmin_{\substack{B, \bm{\alpha}, \{\bm{\lambda}_g\} \\ \text{s.t. } B \in \mathcal{B}_0,  \Gamma^\top \bm{\alpha} = 0}}  \sum_{g=1}^{G}\left\|\hat{\bm{\theta}}_g - (\bm{\alpha} + \Gamma \bm{\lambda}_g + B W_{g})\right\|_2^2 \omega_g.
\end{equation}
Despite the unbiasedness of the calculated estimates (conditionally mean-zero error $\bm{\varepsilon}_g$), $\hat{B}^{MD}$ is generally inconsistent in the asymptotic regime where $n_g$ is fixed. The core issue is sample selection. The indicator $\omega_g$ is a random variable that depends on the specific sample drawn from $F_g$. Since the policy $W_g$ causally affects the distribution $F_g$, it can systematically influence the probability of successful estimation, $\mathbb{P}_{F_g}[\omega_g=1]$. If the policy makes identification easier or harder---e.g., if childcare provision ($W_g$) increases birth rates, thereby reducing the likelihood of observing zero births---then the sample of "identifiable" groups is no longer representative.

Formally, we can decompose the estimator into a component driven by estimation error (which vanishes asymptotically) and a component driven by selection: $\hat{B}^{MD} = \hat{B}_0^{MD} + \hat{B}_1^{MD}$. The estimation error component $\hat{B}_0^{MD}$ converges to zero because $\bm{\varepsilon}_g$ is orthogonal to the policy conditional on selection. The potential bias stems entirely from $\hat{B}_1^{MD}$, which solves the oracle problem on the selected sample:
\begin{equation*}
   (\hat {B}_1^{MD},\hat{\bm{\alpha}}_1^{MD}, \{\hat{\bm{\lambda}}^{MD}_{1,g}\}) := \argmin_{\substack{B, \bm{\alpha}, \{\bm{\lambda}_g\} \\ \text{s.t. } B \in \mathcal{B}_0,  \Gamma^\top \bm{\alpha} = 0}}  \sum_{g=1}^{G}\left\|\bm{\theta}_g - (\bm{\alpha} + \Gamma \bm{\lambda}_g + B W_{g})\right\|_2^2 \omega_g.
\end{equation*}
This objective is structurally identical to the GMM problem analyzed in Section~\ref{sec:identification}, with an effective weighting matrix $\tilde{A}_g = \omega_g I_k$. Because $n_g$ is fixed, $\omega_g$ remains random even as $G \to \infty$. Consequently, Proposition~\ref{prop:consistency} applies directly.

\begin{corollary}[Consistency of MD with Fixed $n_g$] \label{cor:md_fixed_n_consistency}
Suppose Assumption \ref{as:simple_model} holds with $n_g$ being i.i.d. across groups, and $\E[n_g] < \infty$.. Define the population intercept adjusted for selection:
\begin{equation*}
\bm{\alpha}^0 := \argmin_{\bm{\alpha}: \Gamma^\top \bm{\alpha} = 0} \E\left[\omega_g \min_{\bm{\lambda}_g}\|\bm{\alpha}_g - \bm{\alpha} - \Gamma \bm{\lambda}_g\|_2^2\right]
\end{equation*}
and the associated residual heterogeneity $\bm{\epsilon}_g^0$. The feasible MD estimator $\hat{B}^{MD}$ converges in probability to $B_0$ as $G \to \infty$ if and only if
\begin{equation} \label{eq:md_consistency_condition}
      M_{\Gamma}\Cov[\omega_g \bm{\epsilon}_g^0, W_g] = \mathbf{0}_k.
\end{equation}
\end{corollary}

Corollary \ref{cor:md_fixed_n_consistency} highlights the fragility of the MD strategy when the share of discarded groups is significant. Even if $W_g$ is randomly assigned relative to the parameters $\bm{\theta}_g$, consistency fails if $W_g$ influences the feasibility of estimating those parameters. This "feasibility bias" does not vanish as $G \to \infty$ unless group sizes $n_g$ also increase.

\subsection{Inference for the MD Estimator}
\label{sec:md_inference}

Section \ref{sec:selection} established that when the share of discarded groups is substantial, the MD estimator suffers from selection bias. However, this result relies on an asymptotic approximation where $n_g$ is fixed. This approximation is only relevant if the bias component is large relative to the standard error. In many applications, we can construct $\hat{\bm{\theta}}_g$ for the vast majority of groups. Intuitively, if the share of discarded groups is sufficiently small, the bias should be negligible compared to the estimation noise. This section formalizes this intuition.

\subsubsection{Vanishing Selection Bias}
\label{sec:md_bias_bound}

We quantify the relationship between the bias and the share of discarded groups by bounding the deviation $\Delta B = \hat{B}_1^{MD} - B^{\star}$ between the feasible estimator (on the selected sample) and the infeasible oracle.

\begin{prop}[Bias Bound] \label{prp:md_bias_bound}
Let $(\hat {B}_1^{MD}, \hat {\bm{\alpha}}_1^{MD}, \{\hat \lambda^{MD}_{1,g}\})$ be the solution using $\bm{\theta}_g$ as the outcome in the weighted regression \eqref{eq:feasib_md_weighted}. Define $\Delta B = \hat {B}_1^{MD} - B^{\star}$ and $\Delta \alpha = \hat {\bm{\alpha}}_1^{MD} - \bm{\alpha}^{\star}$. Define the matrix $M := \frac{1}{G}\sum_{g=1}^G \omega_g (1,W_{g})^\top (1,W_g)$. Then,
\begin{equation*}
 \left\| (\Delta \alpha, \Delta B) \right\|_F \lesssim \left(\frac{\sqrt{1+\max_{g} \|W_g\|_2^2}}{\lambda_{\min}(M)}  \max_g \|\bm{r}^\star_g\|_2 \right) \frac{\|\omega - 1\|_1}{G}
\end{equation*}
where $\bm{r}^\star_g := \bm{\theta}_g - \bm{\alpha}^{\star} - \Gamma \bm{\lambda}^\star_g - B^{\star} W_g$ is the oracle residual for group $g$, $\| \cdot \|_F$ denotes the Frobenius norm, $\| \cdot \|_2$ the Euclidean norm, $\|\omega - 1\|_1 = \sum_{g=1}^G (1-\omega_g)$ is the number of groups discarded due to $\omega_g=0$, and $\lambda_{\min}(\cdot)$ is the smallest eigenvalue.
\end{prop}
The upper bound in Proposition \ref{prp:md_bias_bound} is proportional to the share of discarded groups, $ \frac{\|\omega - 1\|_1}{G}$. The coefficient of proportionality depends on the unobserved oracle residuals and captures the uncertainty attributable to policy variation. In well-behaved problems, this measure is comparable to the asymptotic standard error (i.e., scaled by $\sqrt{G}$). This implies a simple heuristic for applied researchers: the potential bias is likely negligible if the share of discarded groups is small relative to $\frac{1}{\sqrt{G}}$.

\subsubsection{Asymptotic Equivalence and Standard Errors}
\label{sec:md_asymptotic_equivalence}

The bound above suggests an alternative asymptotic regime where the selection bias vanishes: one where the group sizes $n_g$ are allowed to increase with the total number of groups $G$. This corresponds to a setting where the share of discarded groups becomes small relative to the noise.
Crucially, when groups are large enough to eliminate the selection bias, the within-group estimation noise ($\bm{\varepsilon}_g$) also averages out. Consequently, the feasible MD estimator $\hat{B}^{MD}$ behaves asymptotically exactly like the oracle estimator $B^{\star}$ that uses the true parameters $\bm{\theta}_g$.

\begin{prop}[Asymptotic Normality] \label{prop:md_inference}
Suppose $\plim_{G\to\infty} B^\star = B_0$ for some deterministic $B_0$ and the oracle estimator is asymptotically normal, $\sqrt{G}(B^\star - B_0) \xrightarrow{d} N(0, V_{B^\star})$ for some asymptotic variance matrix $V_{B^\star}$. Suppose $\mathbb{P}_{F_g}[\omega_g=0] \lesssim \exp(-c n_g)$ and $\sum_{g=1}^G \exp(-c n_g) \ll \sqrt{G}$. Suppose that conditionally on $\{\omega_h, \{\tilde D_{h,i}\}, W_h\}_{h=1}^G$, the error terms $\bm{\varepsilon}_g$ are independent across $g$ and the variance $\|n_g\E[\bm{\varepsilon}_g \bm{\varepsilon}_g^\top \mid \omega_g=1, \{\tilde D_{g,i}\}, W_g]\|_{F}$ is uniformly bounded. Finally, suppose $\frac{1}{G}\sum_{g=1}^G W_g$ and $\frac{1}{G}\sum_{g=1}^G W_{g}W_{g}^\top$ converge in probability to well-behaved deterministic limits. Then the feasible MD estimator $\hat{B}^{MD}$ defined in \eqref{eq:feasib_md_weighted} satisfies:
\begin{equation*}
    \sqrt{G}(\hat{B}^{MD} - B_0) \xrightarrow{d} N(0, V_{B^\star}).
\end{equation*}
\end{prop}

This equivalence has key implications for estimating the asymptotic variance $V_{B^\star}$ required for constructing confidence intervals and conducting hypothesis tests using $\hat{B}^{MD}$. The result implies that standard second-stage variance estimation approaches remain valid, provided they are appropriate for the oracle estimator.

Consider, for instance, the commonly used Eicker-Huber-White (EHW) robust variance estimator. Its validity for the oracle regression depends on underlying assumptions about the data-generating process (e.g., independence or specific dependence structures across groups $g$). If these assumptions hold, and the EHW estimator applied to the infeasible oracle regression consistently estimates $V_{B^\star}$, then Proposition \ref{prop:md_inference} implies that the feasible EHW estimator, computed using the feasible second-stage residuals $\hat{\bm{r}}_g = \hat{\bm{\theta}}_g - (\hat{\bm{\alpha}}^{MD} + \Gamma \hat{\bm{\lambda}}^{MD}_g + \hat{B}^{MD} W_{g})$, also consistently estimates $V_{B^\star}$.

The first-stage estimation noise embedded within the feasible residuals  does not distort the properties of such standard second-stage variance estimators, provided the groups are large enough. The contribution of this noise is asymptotically negligible for variance estimation, just as it is for the point estimate. Therefore, researchers can employ robust variance estimation techniques developed for linear models (like EHW or potentially its cluster-robust variants) on the second-stage MD regression, contingent only on the appropriateness of the chosen technique for the underlying structure of the oracle problem itself. No explicit correction for the first-stage estimation is needed asymptotically under these conditions.

The flip side of Proposition~\ref{prop:md_inference} is that if researchers are concerned that the MD estimator is too ``noisy'' due to small groups, they should implicitly also be concerned about the ``bias'' problem caused by selection, as the two issues vanish in the same asymptotic regime.

\begin{remark}[No Unbiased Estimation]\label{rem:uncondition_moments}
This result relies on the first-stage estimates $\hat{\bm{\theta}}_g$ being conditionally unbiased. In models without this property (e.g., TSLS), first-stage bias can persist even as selection bias vanishes. In such cases, mitigating the bias requires either bias-correction techniques (e.g., jackknife) or stronger assumptions on $n_g$.
\end{remark}

\subsection{Restoring Consistency with Auxiliary Information}
\label{subsec:design_info}

What if the share of discarded groups is not sufficiently small? In this scenario, the bias identified in Section \ref{sec:selection} dominates, and we cannot rely on the standard MD estimator. However, the problem is not that the parameter $\bm{\theta}_g$ is fundamentally unidentified, but rather that we lack the information to construct a valid estimator for it in every sample. We now show that if external information about the population moment structure is available, consistency can be restored even when the share of small groups is large. 

\paragraph{Intuition:} Consider the Child Penalty model with a Difference-in-Differences specification. As noted earlier, the standard estimator fails if a municipality has no first births in the sample, forcing us to drop the group. As discussed previously, this scenario is increasingly likely as we restrict the analysis to more narrowly defined subpopulations.\footnote{In our empirical application, which relies on administrative data with millions of observations, we cannot construct the parameters for around $5\%$ of the narrowly defined groups.} However, suppose we know the population probability of having a first birth in municipality $g$, denoted $\pi_g = \mathbb{P}[E_{g,i}=1]$, perhaps from administrative census data or a separate fertility model. With this auxiliary information, we can construct a modified estimator that does not require in-sample variation in treatment:
\begin{equation*}
\hat{\tau}_g^{alt} = \frac{1}{n_g} \sum_{i=1}^{n_g} \frac{(E_{g,i} - \pi_g)}{\pi_g(1-\pi_g)} \Delta Y_{g,i}.
\end{equation*}
This estimator takes a familiar Inverse Probability Weighting (IPW) form. Crucially, it is well-defined and unbiased for $\tau_g$ even if every woman in the specific sample $g$ happens to have $E_{g,i}=0$. Because $\hat{\tau}_g^{alt}$ exists for every group, the selection indicator is effectively $\omega_g=1$ for all $g$. The MD estimator based on these estimates averages out the noise across many groups $G$, ensuring consistency for the policy effect $B_0$ without any sample selection bias.

\paragraph{General Solution:} This logic extends to the general framework. The selection problem arises because the sample Jacobian matrix $\hat{H}_{2,g}$ may be singular. If the population matrix $H_{2,g} := \E_{F_g}[h_2(\tilde D_{g,i})]$ is known, estimated from a large auxiliary dataset, or explicitly modeled, we can bypass the sample invertibility condition entirely.

Define the alternative first-stage estimator:
\begin{equation*}
    \hat{\bm{\theta}}^{alt}_g := H_{2,g}^{-1} \left( \frac{1}{n_g}\sum_{i=1}^{n_g} h_{1}(D_{g,i}) \right) = H_{2,g}^{-1} \hat{H}_{1,g}.
\end{equation*}
Since $H_{2,g}$ is a population quantity assumed to be invertible, $\hat{\bm{\theta}}^{alt}_g$ is always well-defined. Furthermore, it is unconditionally unbiased: $\E_{F_g}[\hat{\bm{\theta}}^{alt}_g] = H_{2,g}^{-1} \E_{F_g}[\hat{H}_{1,g}] = \bm{\theta}_g$. The MD estimator constructed using these design-based estimates,
\begin{align*}
    (\hat{B}^{MD,alt}, \hat{\bm{\alpha}}^{MD,alt}, \{\hat{\bm{\lambda}}^{MD,alt}_{g}\}):= \argmin_{\substack{B, \bm{\alpha}, \{\bm{\lambda}_g\} \\ \text{s.t. } B \in \mathcal{B}_0,  \Gamma^\top \bm{\alpha} = 0}}  \sum_{g=1}^{G}\left\|\hat{\bm{\theta}}^{alt}_g - (\bm{\alpha} + \Gamma \bm{\lambda}_g + B W_{g})\right\|_2^2,
\end{align*}
converges in probability to $B_0$ as $G \to \infty$, regardless of the group size $n_g$.

This result demonstrates that the inconsistency of the standard MD estimator is not an inevitable consequence of small samples; it is a consequence of relying on unstable estimated weights. When these weights are fixed using auxiliary information---analogous to design-based inference methods (e.g., \citealp{Arkhangelsky2024c, Borusyak2023})---the hierarchical structure allows for consistent estimation of policy effects even when micro-estimates are extremely noisy.

\section{Empirical application}
\label{sec:empirics}
This section empirically demonstrates the practical consequences of our theoretical findings. We evaluate the impact of the 2005 Dutch childcare expansion on ``child penalty'' measures \citep{Kleven2019a}. This setting is particularly relevant to illustrate the endogenous weighting bias -- the policy plausibly affects both labor market outcomes and fertility decisions. Moreover, this mechanism provides a potential explanation for the conflicting evidence on the effects of such policies.\footnote{See \cite{Kleven2024} Section III.D for a review.}

A natural approach to demonstrating the associated biases would be to directly replicate an existing study. However, because our method requires both aggregate policy variation and individual-level outcomes, such replication is infeasible, as existing studies rely on restricted-access administrative data. Instead, we leverage a different comprehensive administrative dataset covering the entire Dutch population — comprising millions of observations — to conduct a new empirical analysis that closely mirrors the specifications used in the existing literature. 

We show that using a conventional one-step GMM/OLS estimator yields substantially different conclusions compared with our preferred two-step MD approach. While we do not claim that this divergence is driven exclusively by the bias, its economic magnitude demonstrates that the results are highly sensitive to the implicit weighting scheme imposed by the one-step method. Given this sensitivity, we argue that the MD estimator—which permits transparent, exogenous weighting—provides a strictly more robust basis for policy evaluation.

Our analysis proceeds by first defining the model-based outcomes (Section~\ref{sec:measurement}) and the policy intervention (Section~\ref{sec:childcareReform}). We then present a stylized comparison of GMM and MD estimators, followed by a richer MD specification that showcases its advantages (Section~\ref{sec:policy_analysis}).

\subsection{Measuring the Child Penalty as a Model-Based Outcome}
    \label{sec:measurement}

To conduct our empirical analysis, we use administrative data from the Central Bureau of Statistics of the Netherlands (CBS) on the universe of Dutch residents. Different data sources, such as municipal registers or tax records, are matched through unique, anonymized identifiers for individuals or households. Appendix \ref{sec:data} presents the main variables used and sample construction. 

Our analysis employs a dynamic event-study specification building on the DiD examples in Section \ref{sec:framework}. We use data on individual $i$ residing in municipality $g$ in the year of the birth of the first child. For each individual, we observe covariates $X_{g,i}$--gender and birth cohort ($B_{g,i}$)--and the timing of their first childbirth relative to their birth year, $E_{g,i}$ (relative event time). We model $Y_{g,i,t}$ (earnings or employment) within relevant age ranges as:
\begin{equation}\label{eq:main_cp_spec}
\begin{aligned}
&Y_{g,i,t} = \gamma_{g,i} + \delta_{g,t}(X_{g,i}) + \sum_{h \ge h_0} \tau_{g,i,h} \mathbf{1}\{B_{g,i} + E_{g,i} = t - h\} + \varepsilon_{g,i,t},\\
&\mathbb{E}_{F_{g}}[\varepsilon_{g,i,t}|X_{g,i}, E_{g,i}] = 0.
\end{aligned}
\end{equation}
Here $\gamma_{g,i}$ is an individual fixed effect, $\delta_{g,t}(X_{g,i})$ represents covariate-specific group-time effects, and $\tau_{g,i,h}$ is the individual-specific effect at horizon $h$ relative to childbirth (the ``child penalty''). 

Our primary parameter of interest, the group-level average child penalty, is the model-based outcome that is the focus of our paper:
\begin{align*}
\tau_{g,h}(x,e) := \mathbb{E}_{F_{g}}[\tau_{g,i,h}|X_{g,i} = x, E_{g,i} = e].
\end{align*}
This parameter corresponds to $\bm{\theta}_g$ in our general framework. In our analysis we focus on $\tau_{g,h}(x, e)$ for $h \in \{-2,\dots 3\}$ and $e \in \{27, \dots, 33\}$. For our subsequent policy analysis, we assume that the childcare expansion policy (detailed below) does not affect the distribution of baseline covariates $X_{g,i}$ (gender and birth cohort). 

\paragraph{Diagnostic checks}
Before proceeding to the policy evaluation, we examine $\hat \tau_{g,h}(x,e)$, particularly for pre-event horizons, to assess model validity.  A key concern for MD estimators with finite group sizes is the selection problem discussed in Section~\ref{sec:md_estimator}. To assess whether this concern is valid in our context, we plot in Figure \ref{fig:hist_ngx} the histogram of $n_{g}(x)$---the total number of individuals in group $g$ with given gender and birth cohort. This figure shows substantial variation in group sizes, with some groups very small, underscoring that the problems discussed in Section~\ref{sec:md_estimator} are relevant even for administrative datasets with millions of observations. Still, we find that we can construct estimates for 95\% of the $(g,x)$ cells in our data.\footnote{We use an $n_g(x)$-weighted share, because our subsequent analysis uses such weighting. Proposition \ref{prp:md_bias_bound} naturally extends to accommodate this.} The fact that the share of excluded groups is small suggests that the selection bias quantified in the bound from Proposition~\ref{prp:md_bias_bound} is likely to be negligible relative to the estimation error. This provides confidence in the validity of the MD approach. 

Finally, Figure \ref{fig:agg_CP}, which displays the estimated average child penalties by gender and age at first birth, reveals significant heterogeneity but finds limited evidence of anticipation effects (effects for $h<0$ are near zero), supporting the credibility of the event-study design.

\subsection{Institutional background -- Dutch childcare provision}
    \label{sec:childcareReform}
    
\paragraph{The 2005 Dutch Childcare Reform}
The 2005 Dutch Childcare Act fundamentally reformed early childhood care provision, transforming its financing and structure (see Appendix~\ref{sec:childcareAct}). The reform replaced fragmented local subsidies with a national, demand-driven, tripartite funding model (parents, employers, government), substantially reducing out-of-pocket costs for many families. Concurrently, the Act liberalized the supply side, allowing for-profit entry under a streamlined national regulatory framework.  This combination of stimulated demand and relaxed entry conditions triggered rapid growth in formal childcare capacity. The substantial reduction in costs and expansion of supply make it plausible that the reform influenced not only parental labor supply but also the timing and incidence of childbirth ($E_{g,i}$).
    
\paragraph{Childcare index}
We measure local childcare availability using a municipality-level Childcare Capacity Index ($CCI_{g,t}$), defined as the ratio of childcare workers ($J_{g,t}$) to the mandated staffing level \citep{Decree1996}:
\begin{equation}
\label{eq:CCI}
CCI_{g, t} := \frac{J_{g, t}}{\sum_{l=0}^{5} N_{g, t, l} / R_l},
\end{equation}
where $N_{g,t,l}$ is the number of children aged $l \in \{0,\dots,5\}$. A value $CCI_{g,t}=1$ signifies that the childcare workforce exactly meets the regulated staffing requirement. Recognizing that this index reflects an equilibrium outcome, our empirical strategy controls for time trends interacted with baseline demographic demand ($S_g := log(\sum_{l=0}^{5} N_{g, 1999, l} / R_l)$). Identification of treatment intensity then relies on the remaining spatial and temporal variation in $CCI_{g,t}$ following the nationwide policy implementation. Figure \ref{fig:cci_dist} shows the wide variation in $CCI$ across space and time, confirming that the reform had a powerful and heterogeneous impact across municipalities.

\begin{figure}[!htbp]
    \centering
    \caption{Childcare supply expansion}
    \begin{subfigure}[b]{0.7\textwidth}
        \centering
        \caption{Distribution of childcare index by municipality}   
        \label{fig:cci_dist}
        \includegraphics[width=.95\linewidth, keepaspectratio=true]{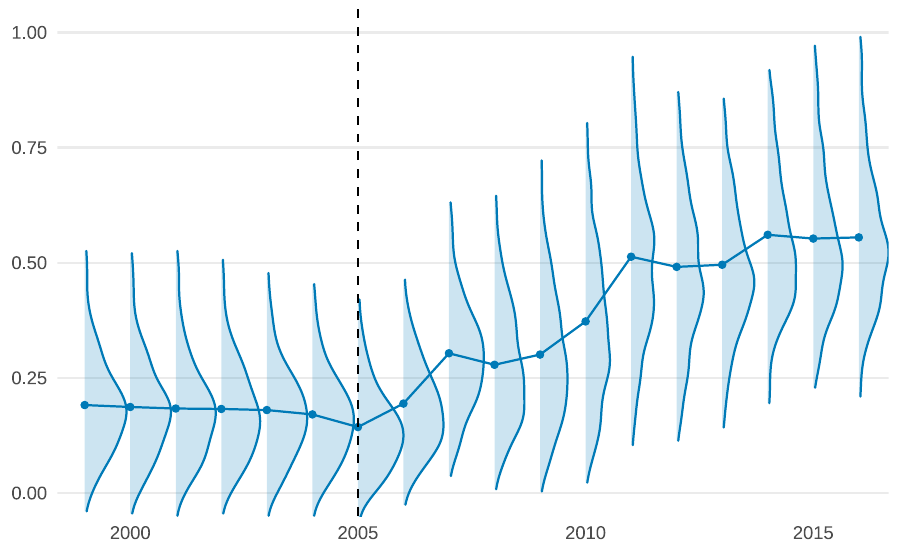} 
    \end{subfigure}
    \hfill
    \begin{subfigure}[b]{0.7\textwidth}   
        \centering
        \caption{Simplified $2 \times 2$ DiD design}   
        \label{fig:cci_2x2}
        \includegraphics[width=0.95\textwidth]{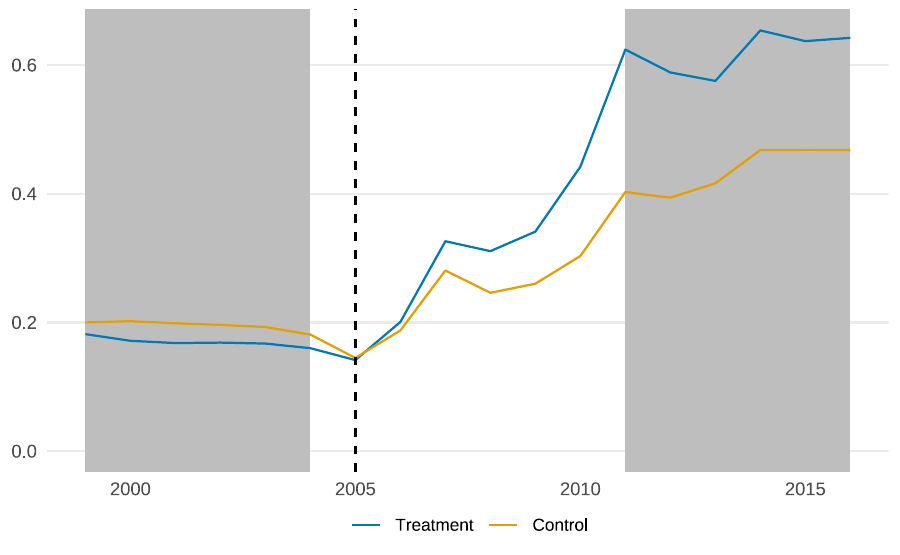}
    \end{subfigure}
    
    \caption*{
        \footnotesize \textit{Notes:} These figures present the variation in childcare supply per preschool-aged children across municipalities from 1999 to 2016. Our childcare capacity index ($CCI$) for each municipality is calculated by dividing the number of childcare jobs in a given municipality $g$ and year $t$ ($N^{jobs}_{g,t}$) by the number of required number of childminders in the same locality (see equation~\eqref{eq:CCI}). The vertical line illustrates the timing of the 2005 Dutch childcare expansion reform. Panel (a) illustrates the substantial variation in childcare availability between different municipalities and the large increase due to the 2005 childcare expansion reform. Dots represent the mean $CCI$ in a given year, whereas the shaded area represents the distribution of $CCI$ across municipalities in that year. Panel (b) illustrates the equivalent simplified $2 \times 2$ DiD design, where the time variation is binary (gray area) and treatment is binary (see Section \ref{sec:empirics_GMM_vs_MD}). Treatment is defined as municipalities with above median expansion of $CCI$ between the baseline period (1999-2004) and the post-expansion period (2011-2016).
    }
    
\end{figure}

\subsection{Policy evaluation}\label{sec:policy_analysis}
Our empirical analysis of the childcare expansion proceeds in two stages. Initially, we present a stylized comparison of GMM and MD approaches. This first exercise serves not only as a practical illustration of the pitfalls discussed in Sections~\ref{sec:identification} and \ref{sec:md_estimator} but also reflects a common empirical strategy, thereby underscoring the relevance of the biases we identify. Subsequently, we undertake a more nuanced MD analysis that directly incorporates our framework's insights regarding policy-induced compositional effects, demonstrating its advantages for applied research.

\subsubsection{Comparison between GMM and MD}
\label{sec:empirics_GMM_vs_MD}
To empirically demonstrate the GMM/OLS challenges identified in Section \ref{sec:identification}, we first examine a common evaluation strategy for policies with continuous local intensity ($CCI_{g,t}$): binarizing the treatment variable. For our illustration, we construct a binary treatment indicator, $W_g$, based on the median expansion in childcare capacity:
\begin{align*}
    W_{g} := \mathbf{1}\{\overline{CCI}_{g}^{\text{post}} - \overline{CCI}_{g}^{\text{pre}} > Med(\overline{CCI}_{g}^{\text{post}} - \overline{CCI}_{g}^{\text{pre}}) \},
\end{align*}
where $\overline{CCI}_{g}^{\text{pre}}$ and $\overline{CCI}_{g}^{\text{post}}$ are average $CCI$ levels in group $g$ during 1999-2004 and 2011-2016, respectively.\footnote{The childcare index is relatively stable in these periods, with most of the expansion happening in $2005-2010$, making them appropriate for defining $W_g$. }

\paragraph{GMM estimator}
Figure \ref{fig:cci_2x2} shows apparent parallel evolution of the average childcare index for the two groups, but also reveals that ``control'' groups ($W_g=0$) still responded to the nationwide policy. This highlights the information loss from binarization. Nonetheless, this simplification allows for a straightforward OLS specification that is common in applied work (e.g., \citealp{Kleven2024,Lim2023,Rabate2021}):
\begin{equation}\label{eq:GMM_CP}
\begin{aligned}
    Y_{g,i,t} &= \gamma_{g,i} + \delta_t + \lambda_{t-B_{g,i}} + \sum_{h \ge h_0 } \alpha_{h}\mathbf{1}\left\{t-B_{g,i} - E_{g,i} = h \right\} \\
&+ \sum_{h \ge h_0} \rho_{h}\mathbf{1}\left\{t-B_{g,i} - E_{g,i} = h \right\}W_{g} \\
&+ \sum_{h \ge h_0} \xi_{h}\mathbf{1}\left\{t-B_{g,i} - E_{g,i} = h \right\}\mathbf{1}\{E_{g,i}+ B_{g,i}>2005\} \\
&+ \sum_{h \ge h_0} \beta_h\mathbf{1}\left\{t-B_{g,i} - E_{g,i} = h \right\} W_{g}\mathbf{1}\{E_{g,i} + B_{g,i}>2005 \} + \nu_{g,i,t}. 
\end{aligned}
\end{equation}
Here, $\beta_h$ captures the policy's effect. This OLS specification is precisely the type of one-step estimator vulnerable to the endogenous weighting bias identified in Section~\ref{sec:identification}. Mechanically, the OLS estimate of $\beta_h$ is a weighted average of group-specific effects, where the implicit weights are determined by the within-group distribution of $E_{g,i}$. If the policy $W_g$ affects the distribution of fertility timing, it alters these implicit weights, creating a spurious correlation between the policy and the weighting matrix that renders the estimator inconsistent.

\paragraph{MD estimator}
An alternative, MD approach relies on the first-step estimates $\hat \tau_{g,h}(x,e)$, relating them directly to the same binarized policy:
\begin{equation}
\label{eq:MD_CP}
    \hat \tau_{g,h}(x,e) = \alpha_h + \rho_h W_g + \xi_{h} \mathbf{1}\{ e + b > 2005\} + \beta_h \mathbf{1}\{ e + b > 2005\}W_g + \upsilon_{g,h}(x,e).
\end{equation}
Here, $x$ includes covariates (education and birth cohort $b$). We run this regression separately for men and women, weighting by $n_g(x)$.
Unlike the implicit and potentially endogenous weights in the GMM specification, here we apply explicit weights, $n_g(x)$. This choice of exogenous weights is the key feature of the MD approach that mechanically purges the endogenous weighting bias.

\paragraph{Comparing empirical results}
Figure~\ref{fig:GMMvsMD_cciEffect} reveals a large difference between the two approaches. The conventional one-step GMM/OLS estimator (Panel a) suggests the policy had large effects, increasing mothers' post-childbirth earnings by up to 12,000 euros and labor force participation by 13 percentage points. In contrast, the MD approach (Panel b), which is robust to the endogenous weighting bias, finds that these effects are much smaller and statistically indistinct from the effects for men.

This divergence is a direct empirical manifestation of the theoretical problem identified in this paper. The one-step GMM/OLS specification \eqref{eq:GMM_CP} implicitly uses the conditional distribution of fertility timing ($E_{g,i}$) to weigh observations, and this distribution may be affected by the policy. While this is a population-level concern, the problem is exacerbated in practice by finite-sample issues. For the many groups with a small number of individuals $n_g(x)$ (Figure~\ref{fig:hist_ngx}), the observed distribution of $E_{g,i}$ can be highly variable due to random sampling alone. Not all of this variation in the implicit weights can or should be attributed to a causal effect of the policy, but the GMM/OLS estimator uses it regardless, without raising any warnings.

This is precisely why the MD specification is more attractive empirically. It forces the researcher to confront this variation---whether it stems from policy-induced effects or idiosyncratic noise---and make an explicit choice about weighting. By using pre-determined weights like $n_g(x)$, the MD estimator \eqref{eq:MD_CP} avoids relying on potentially noisy, endogenous, and data-driven implicit weights. The fact that this single methodological choice changes the results so profoundly underscores the practical importance of our recommendation for transparent, two-step estimation.

\begin{remark}[On Normalizing Child Penalties]
The child penalty literature often normalizes estimates (e.g., by baseline earnings). While entirely feasible within our framework, such normalization would produce an additional finite-sample bias discussed in Remark~\ref{rem:uncondition_moments}. As a result, we opt to analyze unnormalized child penalties, isolating the core methodological issues of GMM endogenous weighting from important but separate statistical problems. 
\end{remark}

\begin{figure}[!htbp]
    \centering
    \caption{GMM vs MD: Effect of the childcare provision expansion on CP}
    \label{fig:GMMvsMD_cciEffect}
    \begin{subfigure}[b]{0.85\textwidth}
        \centering
        \caption{GMM}
        \label{fig:GMM_cciEffect}
        \includegraphics[width=.95\linewidth, keepaspectratio=true]{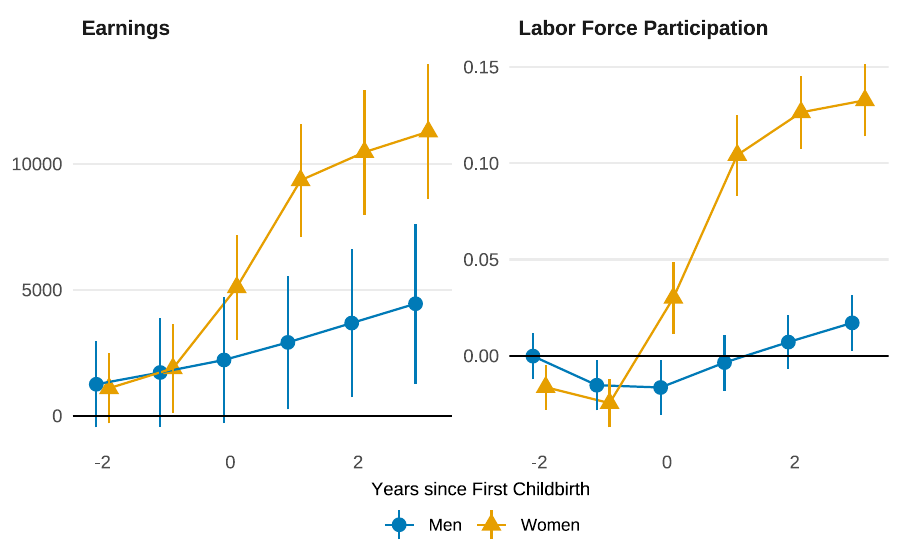}
    \end{subfigure}
    \vskip\baselineskip
    \begin{subfigure}[b]{0.85\textwidth}
        \centering
        \caption{MD}
        \label{fig:MD_cciEffect}
        \includegraphics[width=.95\linewidth, keepaspectratio=true]{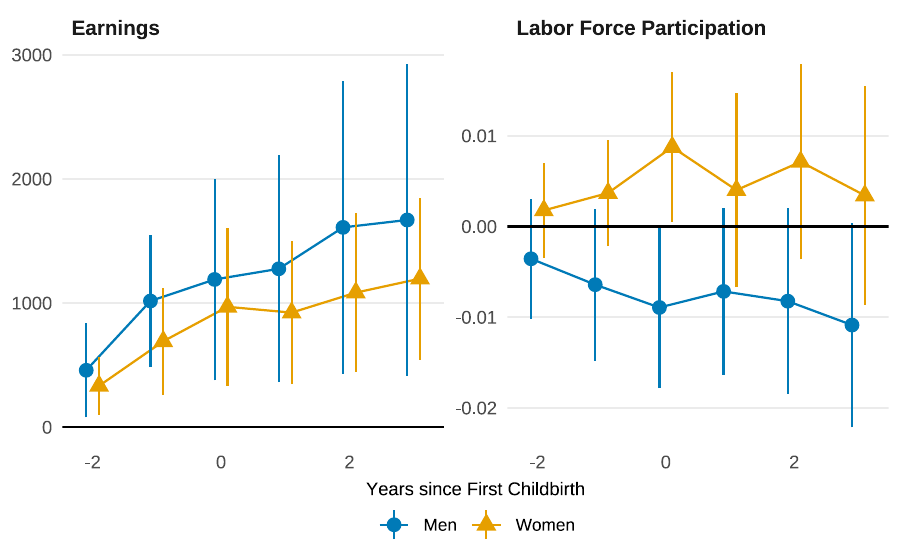}
    \end{subfigure}
     \caption*{
        \footnotesize \textit{Notes:} This figure presents the effect of the childcare provision expansion on child penalties (CP) in earnings, estimated separately using the GMM and the MD approaches (see Section \ref{sec:empirics_GMM_vs_MD} for more details). Figure \ref{fig:GMM_cciEffect} presents the results of the GMM approach commonly used in the literature (specification \eqref{eq:GMM_CP}). Figure \ref{fig:MD_cciEffect} presents the results of the MD approach we suggest (specification \eqref{eq:MD_CP}). We split the estimation between men (blue) and women (orange) and present the results for both the intensive margin (earnings, as shown in the panel above) and the extensive margin (participation, as shown in the panel below). Each dot presents the corresponding coefficient and its marginal 95\% confidence interval based on standard errors clustered by municipality $g$.
    }
\end{figure}

\subsubsection{A Richer MD Specification for Policy Evaluation}
\label{sec:rich_md}

Moving beyond the stylized comparison, we now leverage the flexibility of the MD approach to conduct a more nuanced evaluation. This allows us to use the continuous nature of the policy intensity, $CCI_{g,t}$, and investigate the compositional effects discussed in Section \ref{sec:aggregate_interpretation}. We estimate the following specification, using $n_{g}(x)$ weights:
\begin{align}
    \hat \tau_{g,h}(x,e) = \alpha_{g,h}(x) + \lambda^{(0)}_h(x,e) + \lambda^{(1)}_h(x,e)S_g + \sum_{j=-2}^{h}\beta_{h,j} CCI_{g,b +e +j} + \upsilon_{g,h}(x,e).
    \label{eq:rich_md}
\end{align}
To control flexibly for a wide array of potential confounders this specification allows for granular fixed effects by population group ($x$): group and time-since-first-birth ($\alpha_{g,h}(x)$), differential time fixed effects by first-birth-year ($\lambda^{(0)}_h(x,e)$), and time trends by first-birth-year interacted with baseline demographic demand ($\lambda^{(1)}_h(x,e)S_g$). The key term of interest is the summation, which captures the rich, dynamic effects of the policy. 
Crucially, incorporating $CCI$ levels from periods prior to childbirth ($j<0$) provides a powerful specification test. It allows us to probe for policy-induced compositional shifts that might manifest as pre-event ``effects'' on the average child penalty. If $\beta_{h,j}$ for $j<0$ were significant, it may suggest that changes in childcare capacity are correlated with our outcome even before they can have a direct causal effect, likely via changes in who becomes a mother in a given municipality.

Table \ref{tab:reg_flex_f} presents the estimated effects for mothers. We find no evidence of significant effects for $j<0$; the coefficients are small and statistically insignificant. This null result suggests that the compositional confounding is not a first-order issue in this richer specification. For the post-birth periods, we find that an expansion of childcare capacity in the year of childbirth has a persistent positive impact on mothers’ labor force participation (the extensive margin). We also observe positive effects on mothers' earnings from expansions one year after birth, but none for the extensive margin, suggesting some role for the intensive margin when formal maternity leave ends. These patterns are not mirrored for fathers (Appendix Table \ref{tab:reg_flex_m}), whose participation and earnings are largely unchanged, consistent with strong labor-market attachment of fathers commonly found in the literature. However, fathers' earnings respond positively to contemporaneous increases in childcare capacity, suggesting an intensive margin response. These plausible, nuanced results showcase the added value of the MD approach's flexibility in uncovering dynamic effects that would have been hard to assess in the one-step GMM approach.

\begin{table}[!htbp]
    \centering
    \caption{Rich MD specification -- the childcare provision expansion on CP (mothers)}
    \label{tab:reg_flex_f}
    \begin{minipage}{.7\linewidth}
        \subcaption{Earnings}
        \label{tab:reg_flex_earn_f}
        \resizebox{\columnwidth}{!}{%
            %\resizebox{\tabwidth}{!}{
\resizebox{\ifdim\width>\linewidth 1\linewidth\else\width\fi}{!}{
\begin{talltblr}[         %% tabularray outer open
entry=none,label=none,
note{}={+ p < 0.1, * p < 0.05, ** p < 0.01},
]                     %% tabularray outer close
{                     %% tabularray inner open
colspec={Q[]Q[]Q[]Q[]Q[]Q[]Q[]},
column{2,3,4,5,6,7}={}{halign=c,},
column{1}={}{halign=l,},
hline{14}={1,2,3,4,5,6,7}{solid, black, 0.05em},
}                     %% tabularray inner close
\toprule
& $\hat{\tau}_{g, -2}$ & $\hat{\tau}_{g, -1}$ & $\hat{\tau}_{g, 0}$ & $\hat{\tau}_{g, 1}$ & $\hat{\tau}_{g, 2}$ & $\hat{\tau}_{g, 3}$ \\ \midrule %% TinyTableHeader
$CCI_{g, b+e-2}$ & 283.5 & 384.8 & 798.3 & 355.2 & 398.6 & -25.3 \\
& (403.1) & (631.4) & (602.3) & (714.4) & (735.4) & (895.9) \\
$CCI_{g, b+e-1}$ &  & 735.5 & 778.0 & 866.9 & 838.7 & 1091.3 \\
&  & (544.3) & (541.7) & (753.6) & (778.7) & (690.2) \\
$CCI_{g, b+e}$ &  &  & 394.9 & 191.7 & 109.2 & 287.8 \\
&  &  & (485.9) & (485.3) & (716.6) & (803.4) \\
$CCI_{g, b+e+1}$ &  &  &  & 1510.9* & 1522.9* & 1969.4* \\
&  &  &  & (618.4) & (634.0) & (892.2) \\
$CCI_{g, b+e+2}$ &  &  &  &  & 407.1 & -59.1 \\
&  &  &  &  & (630.2) & (635.2) \\
$CCI_{g, b+e+3}$ &  &  &  &  &  & 1330.8 \\
&  &  &  &  &  & (889.7) \\
N & 10,941 & 10,941 & 10,941 & 10,941 & 10,941 & 10,941 \\
$R^2$ & 0.163 & 0.258 & 0.314 & 0.369 & 0.385 & 0.369 \\
FE: Municipality $g$ & X & X & X & X & X & X \\
FE: $B_{g, i} \times E_{g, i}$ & X & X & X & X & X & X \\
FE: $(B_{g, i} \times E_{g, i})S_g$ & X & X & X & X & X & X \\
\bottomrule
\end{talltblr}
}

        }
    \end{minipage}

    \begin{minipage}{.7\linewidth}
        \subcaption{Participation}
        \label{tab:reg_flex_particip_f}
        \resizebox{\columnwidth}{!}{%
            %\resizebox{\tabwidth}{!}{
\resizebox{\ifdim\width>\linewidth 1\linewidth\else\width\fi}{!}{
\begin{talltblr}[         %% tabularray outer open
entry=none,label=none,
note{}={+ p < 0.1, * p < 0.05, ** p < 0.01},
]                     %% tabularray outer close
{                     %% tabularray inner open
colspec={Q[]Q[]Q[]Q[]Q[]Q[]Q[]},
column{2,3,4,5,6,7}={}{halign=c,},
column{1}={}{halign=l,},
hline{14}={1,2,3,4,5,6,7}{solid, black, 0.05em},
}                     %% tabularray inner close
\toprule
& $\hat{\tau}_{g, -2}$ & $\hat{\tau}_{g, -1}$ & $\hat{\tau}_{g, 0}$ & $\hat{\tau}_{g, 1}$ & $\hat{\tau}_{g, 2}$ & $\hat{\tau}_{g, 3}$ \\ \midrule %% TinyTableHeader
$CCI_{g, b+e-2}$ & -0.012 & -0.026+ & -0.007 & -0.006 & -0.021 & -0.018 \\
& (0.010) & (0.014) & (0.015) & (0.017) & (0.019) & (0.020) \\
$CCI_{g, b+e-1}$ &  & 0.027* & 0.003 & 0.005 & -0.002 & -0.006 \\
&  & (0.011) & (0.011) & (0.013) & (0.017) & (0.016) \\
$CCI_{g, b+e}$ &  &  & 0.031** & 0.031* & 0.031* & 0.038* \\
&  &  & (0.011) & (0.012) & (0.014) & (0.016) \\
$CCI_{g, b+e+1}$ &  &  &  & 0.009 & -0.000 & -0.007 \\
&  &  &  & (0.012) & (0.016) & (0.017) \\
$CCI_{g, b+e+2}$ &  &  &  &  & 0.005 & 0.001 \\
&  &  &  &  & (0.014) & (0.012) \\
$CCI_{g, b+e+3}$ &  &  &  &  &  & 0.006 \\
&  &  &  &  &  & (0.015) \\
N & 10,941 & 10,941 & 10,941 & 10,941 & 10,941 & 10,941 \\
$R^2$ & 0.050 & 0.056 & 0.068 & 0.109 & 0.121 & 0.128 \\
FE: Municipality $g$ & X & X & X & X & X & X \\
FE: $B_{g, i} \times E_{g, i}$ & X & X & X & X & X & X \\
FE: $(B_{g, i} \times E_{g, i})S_g$ & X & X & X & X & X & X \\
\bottomrule
\end{talltblr}
}

        }
    \end{minipage}

    \vspace{0.5em}
    \caption*{
        \footnotesize \textit{Notes:} These tables present the effect of the childcare provision expansion on child penalties (CP) of mothers in earnings (above) and labor force participation (below). See Section \ref{sec:rich_md} for more details and Equation~\eqref{eq:rich_md} for the specification. Standard errors clustered by municipality $g$ are in parentheses.
     }
\end{table}

\section{Conclusion}
\label{sec:conclusion}
This paper develops a framework for analyzing causal effects of group-level policies on model-based outcomes, defined via micro-level moments. We highlight how policy-induced changes in the distribution of the micro-level data can bias common estimators. We find that one-step GMM estimators, including OLS with policy interactions, generally yield inconsistent estimates due to an endogenous weighting problem. Two-stage MD estimators, while avoiding this specific issue, can still be inconsistent if the groups are small. Our analysis, however, shows that the MD estimator is consistent when the share of discarded groups for which the first-step estimator is undefined is sufficiently small. Alternatively, consistency can be restored by using auxiliary data, such as known population moments. Future research could extend our framework to nonlinear moment conditions, explore random effects models in settings with small groups where MD estimators face challenges, and investigate models with complex interdependencies across groups, allowing for the analysis of network data.

\clearpage
\bibliography{sections/references}
\bibliographystyle{econ-aea}

%\clearpage
%\section*{Figures and tables}
%\label{sec:figures}
%\input{sections/figAndTable}

\clearpage

%%%%%%%%%%%% APPENDIX %%%%%%%%%%%%%%%%%%%%%%
\begin{center}
    \Large Supplementary Appendix 
\end{center}
\appendix

\counterwithin{figure}{section}
\counterwithin{table}{section}
\counterwithin{equation}{section}

\section{Proofs}
    \label{app:theory}
    First, we state the necessary algebraic restriction on the oracle that ensures that the estimator is well-defined:
\begin{assumption}\label{as:identification_subspace}
Let $\mathcal{B}_0 \subset \mathbb{R}^{k \times p}$ be the linear subspace of feasible policy effects. Let $P_{\Gamma^\perp} = I_k - \Gamma(\Gamma^\top \Gamma)^{-1}\Gamma^\top$ be the orthogonal projection matrix onto the null space of $\Gamma^\top$. We assume that the projection map restricted to $\mathcal{B}_0$ is injective. Specifically, let
\begin{equation*}
    \kappa := \inf_{B \in \mathcal{B}_0, B \neq \mathbf{0}} \frac{\| P_{\Gamma^\perp} B \|_F}{\| B \|_F}.
\end{equation*}
We assume $\kappa > 0$. The constant $\kappa$ represents the minimal singular value of the restricted projection and measures the preservation of the signal after removing group fixed effects.
\end{assumption}

We now explicitly verify Assumption~\ref{as:identification_subspace} for the example in Section~\ref{subsec:formal_setup} where $\Gamma$ describes group-level fixed effects and the policy effect $B$ is homogeneous across outcomes.

\begin{lemma}[Identification with Scalar Effects]
Suppose the dimension of the outcome is $k > 1$. Let $\Gamma = \bm{\iota}_k$ be a $k \times 1$ vector of ones, and let $\mathcal{B}_0 = \{ \beta I_k : \beta \in \mathbb{R} \}$ be the subspace of scalar matrices. Then Assumption~\ref{as:identification_subspace} holds with the constant $\kappa = \sqrt{\frac{k-1}{k}}$.
\end{lemma}

\begin{proof}
The orthogonal projection matrix onto the null space of $\Gamma^\top$ is the de-meaning matrix (or centering matrix):
\begin{align*}
    P_{\Gamma^\perp} = I_k - \Gamma(\Gamma^\top \Gamma)^{-1}\Gamma^\top = I_k - \frac{1}{k}\bm{\iota}_k\bm{\iota}_k^\top.
\end{align*}
Consider an arbitrary non-zero element $B \in \mathcal{B}_0$, which can be written as $B = \beta I_k$ for some scalar $\beta \neq 0$. The Frobenius norm of $B$ is:
\begin{align*}
    \|B\|_F = \sqrt{\text{Tr}(B^\top B)} = \sqrt{\text{Tr}(\beta^2 I_k)} = |\beta| \sqrt{k}.
\end{align*}
Next, we compute the projection of $B$. Since $P_{\Gamma^\perp} B = P_{\Gamma^\perp} (\beta I_k) = \beta P_{\Gamma^\perp}$, the Frobenius norm of the projected matrix is:
\begin{align*}
    \|P_{\Gamma^\perp} B\|_F = |\beta| \|P_{\Gamma^\perp}\|_F.
\end{align*}
Since $P_{\Gamma^\perp}$ is a projection matrix, it is idempotent ($P_{\Gamma^\perp} P_{\Gamma^\perp} = P_{\Gamma^\perp}$), so its squared Frobenius norm equals its trace:
\begin{align*}
    \|P_{\Gamma^\perp}\|_F^2 = \text{Tr}(P_{\Gamma^\perp}^\top P_{\Gamma^\perp}) = \text{Tr}(P_{\Gamma^\perp}) = \text{Tr}\left(I_k - \frac{1}{k}\bm{\iota}_k\bm{\iota}_k^\top\right) = k - 1.
\end{align*}
Thus, $\|P_{\Gamma^\perp} B\|_F = |\beta| \sqrt{k-1}$.
Finally, we compute the ratio $\kappa$:
\begin{align*}
    \frac{\|P_{\Gamma^\perp} B\|_F}{\|B\|_F} = \frac{|\beta|\sqrt{k-1}}{|\beta|\sqrt{k}} = \sqrt{\frac{k-1}{k}}.
\end{align*}
Since $k > 1$, we have $\kappa > 0$, confirming that the mapping is injective.
\end{proof}

\subsection{Proof of Proposition \ref{prop:consistency}}

\begin{proof}
The GMM estimator $(\hat{\bm{\alpha}}, \hat{B})$ minimizes the objective function:
\begin{align*}
    L(\bm{\alpha}, B) = \sum_{g=1}^G (\bm{\theta}_g - \bm{\alpha} - B W_g)^\top \tilde{A}_g (\bm{\theta}_g - \bm{\alpha} - B W_g)
\end{align*}
subject to $\Gamma^\top \bm{\alpha} = \mathbf{0}$ and $B \in \mathcal{B}_0$. We assume that the subspace $\mathcal{B}_0$ satisfies the identification condition that $\mathcal{B}_0 \cap \im(\Gamma) = \{\mathbf{0}\}$, implying that any $B \in \mathcal{B}_0$ is uniquely determined by its projection onto $\Gamma^\perp$.

Since $\tilde{A}_g$ is symmetric and positive definite, the matrix $\check{A}_g:=U^\top \tilde{A}_g U$ is also positive definite, where $U$ is a $k \times k'$ matrix with orthonormal columns ($U^\top U = I_{k'}$) forming a basis for $\im(P_{\Gamma^\perp})$. 
We transform the problem into the lower-dimensional space spanned by $U$. Let $\tilde{\bm{\alpha}} := U^\top \bm{\alpha}$ and $\tilde{B} := U^\top B$. Due to Assumption~\ref{as:identification_subspace}, the mapping $B \mapsto \tilde{B}$ is a bijection between $\mathcal{B}_0$ and its image $\tilde{\mathcal{B}}_0 \subset \mathbb{R}^{k' \times p}$.

We first derive the solution for the ``unconstrained'' problem where $\tilde{B}$ can be any matrix in $\mathbb{R}^{k' \times p}$ (corresponding to $B$ ranging over the entire subspace $\Gamma^\perp$). Let $(\hat{\tilde{\bm{\alpha}}}^{unc}, \hat{\tilde{B}}^{unc})$ denote this unconstrained estimator. It satisfies the normal equations:
\begin{align}
\left(\sum_{g=1}^G \check{A}_g\right) \hat{\tilde{\bm{\alpha}}}^{unc} + \left(\sum_{g=1}^G \check{A}_g \hat{\tilde{B}}^{unc} W_g\right) &= \sum_{g=1}^G U^\top \tilde{A}_g \bm{\theta}_g \label{eq:norm_eq1_detailed_final} \\
\left(\sum_{g=1}^G W_g \otimes \check{A}_g\right) \hat{\tilde{\bm{\alpha}}}^{unc} + \left(\sum_{g=1}^G W_g W_g^\top \otimes \check{A}_g\right) \vectorize(\hat{\tilde{B}}^{unc}) &= \vectorize\left(\sum_{g=1}^G U^\top \tilde{A}_g \bm{\theta}_g W_g^\top\right) \label{eq:norm_eq2_detailed_final}
\end{align}
Under Assumption~\ref{as:simple_model}(c), the potential outcome is $\bm{\theta}_g(W_g) = \bm{\alpha}_g + B_0 W_g$, with $B_0 \in \mathcal{B}_0$. We define the transformed true parameter as $\tilde{B}_0 = U^\top B_0$. Similarly, the population intercept $\bm{\alpha}^0$ satisfies $\Gamma^\top \bm{\alpha}^0 = \mathbf{0}$, so $\bm{\alpha}^0 = U \tilde{\bm{\alpha}}^0$ for some $\tilde{\bm{\alpha}}^0$. We define the residual $\bm{\epsilon}_g^0 = \bm{\alpha}_g - \bm{\alpha}^0$. By the definition of $\bm{\alpha}^0$ as the GLS population intercept, it satisfies $\E[\check{A}_g(U^\top\bm{\alpha}_g - \tilde{\bm{\alpha}}^0)] = \mathbf{0}$, which implies $\E[U^\top \tilde{A}_g \bm{\epsilon}_g^0] = \mathbf{0}$.

We substitute the projected model $U^\top \bm{\theta}_g = U^\top \bm{\alpha}^0 + U^\top \bm{\epsilon}_g^0 + \tilde{B}_0 W_g$ into the right-hand side (RHS) of the normal equations:
\begin{align*}
    \text{RHS of \eqref{eq:norm_eq1_detailed_final}} &= \sum U^\top \tilde{A}_g \bm{\epsilon}_g^0 + \left(\sum \check{A}_g\right) \tilde{\bm{\alpha}}^0 + \left(\sum \check{A}_g \tilde{B}_0 W_g\right),\\
       \text{RHS of \eqref{eq:norm_eq2_detailed_final}} &= \vectorize\left(\sum U^\top \tilde{A}_g \bm{\epsilon}_g^0 W_g^\top\right) + \vectorize\left(\sum \check{A}_g \tilde{\bm{\alpha}}^0 W_g^\top\right) + \vectorize\left(\sum \check{A}_g \tilde{B}_0 W_g W_g^\top\right).
\end{align*}
Using the identity $\vectorize(A B C) = (C^\top \otimes A) \vectorize(B)$, note that $\vectorize(\check{A}_g \tilde{\bm{\alpha}}^0 W_g^\top) = (W_g \otimes \check{A}_g)\tilde{\bm{\alpha}}^0$ and $\vectorize(\check{A}_g \tilde{B}_0 W_g W_g^\top) = (W_g W_g^\top \otimes \check{A}_g)\vectorize(\tilde{B}_0)$.

Let $\Delta\tilde{\bm{\alpha}}^{unc} = \hat{\tilde{\bm{\alpha}}}^{unc} - \tilde{\bm{\alpha}}^0$ and $\Delta\tilde{B}^{unc} = \hat{\tilde{B}}^{unc} - \tilde{B}_0$. Substituting these into the left-hand side (LHS) of the normal equations and equating LHS to RHS causes the terms involving $\tilde{\bm{\alpha}}^0$ and $\tilde{B}_0$ to cancel out. This leaves a system for the estimation error:
\begin{align*}
\left(\sum_{g=1}^G \check{A}_g\right) \Delta\tilde{\bm{\alpha}}^{unc} + \left(\sum_{g=1}^G W_g^\top \otimes \check{A}_g\right) \vectorize(\Delta\tilde{B}^{unc}) &= \sum_{g=1}^G U^\top \tilde{A}_g \bm{\epsilon}_g^0 \\
\left(\sum_{g=1}^G W_g \otimes \check{A}_g\right) \Delta\tilde{\bm{\alpha}}^{unc} + \left(\sum_{g=1}^G W_g W_g^\top \otimes \check{A}_g\right) \vectorize(\Delta\tilde{B}^{unc}) &= \vectorize\left(\sum_{g=1}^G U^\top \tilde{A}_g \bm{\epsilon}_g^0 W_g^\top\right)
\end{align*}
Let $\tilde{H}_{11} = \sum \check{A}_g$, $\tilde{H}_{12} = \sum (W_g^\top \otimes \check{A}_g)$, $\tilde{H}_{21} = \sum (W_g \otimes \check{A}_g)$, and $\tilde{H}_{22} = \sum (W_g W_g^\top \otimes \check{A}_g)$. Let $C_1^{\epsilon} = \sum U^\top \tilde{A}_g \bm{\epsilon}_g^0$ and $C_2^{\epsilon} = \vectorize(\sum U^\top \tilde{A}_g \bm{\epsilon}_g^0 W_g^\top)$. Solving for $\vectorize(\Delta\tilde{B}^{unc})$ using the Schur complement $\tilde{S} = \tilde{H}_{22} - \tilde{H}_{21}\tilde{H}_{11}^{-1}\tilde{H}_{12}$:
\begin{align*}
\vectorize(\Delta\tilde{B}^{unc}) = \tilde{S}^{-1} \left( C_2^{\epsilon} - \tilde{H}_{21}\tilde{H}_{11}^{-1} C_1^{\epsilon} \right).
\end{align*}
The feasible estimator $\hat{\tilde{B}}$ minimizes the quadratic objective restricted to $\tilde{B} \in \tilde{\mathcal{B}}_0$. This is equivalent to minimizing the distance to the unconstrained estimator in the metric of the concentrated Hessian $\tilde{S}$:
\begin{align*}
\vectorize(\hat{\tilde{B}}) = \argmin_{b \in \vectorize(\tilde{\mathcal{B}}_0)} (b - \vectorize(\hat{\tilde{B}}^{unc}))^\top \tilde{S} (b - \vectorize(\hat{\tilde{B}}^{unc}))
\end{align*}
Consequently, the estimation error $\vectorize(\Delta\tilde{B}) = \vectorize(\hat{\tilde{B}} - \tilde{B}_0)$ is the orthogonal projection of the unconstrained error $\vectorize(\Delta\tilde{B}^{unc})$ onto the subspace $\vectorize(\tilde{\mathcal{B}}_0)$ under the inner product defined by $\tilde{S}$. Let $\mathcal{P}_{\tilde{\mathcal{B}}_0, \tilde{S}}$ denote this projection operator. Then:
\begin{align*}
\vectorize(\Delta\tilde{B}) = \mathcal{P}_{\tilde{\mathcal{B}}_0, \tilde{S}} \left[ \tilde{S}^{-1} \left( C_2^{\epsilon} - \tilde{H}_{21}\tilde{H}_{11}^{-1} C_1^{\epsilon} \right) \right].
\end{align*}
We now consider the probability limits as $G \to \infty$. Under Assumption~\ref{as:simple_model}, the sample averages converge in probability to their population counterparts:
$\frac{1}{G}\tilde{H}_{ij} \xrightarrow{p} H_{ij,plim}$, $\frac{1}{G}C_1^{\epsilon} \xrightarrow{p} \E[U^\top \tilde{A}_g \bm{\epsilon}_g^0] = \mathbf{0}$, and $\frac{1}{G}C_2^{\epsilon} \xrightarrow{p} \vectorize(\E[U^\top \tilde{A}_g \bm{\epsilon}_g^0 W_g^\top])$.

The probability limit of the scaled Schur complement is 
\begin{align*}
    S_{plim} = H_{22,plim} - H_{21,plim}H_{11,plim}^{-1}H_{12,plim}.
\end{align*}
This matrix is positive definite if the full Hessian is positive definite, which is guaranteed by Assumption~\ref{as:simple_model}(a) ($\Var[W_g]$ is positive definite). The projection operator converges to the population projection $\mathcal{P}_{\tilde{\mathcal{B}}_0, S_{plim}}$. Thus:
\begin{align*}
\plim_{G\to\infty} \vectorize(\Delta\tilde{B}) &= \mathcal{P}_{\tilde{\mathcal{B}}_0, S_{plim}} \left[ S_{plim}^{-1} \vectorize(\E[U^\top \tilde{A}_g \bm{\epsilon}_g^0 W_g^\top]) \right].
\end{align*}
Consistency for the target parameter $B_0$ requires $\plim \hat{B} = B_0$, which by Assumption~\ref{as:identification_subspace} is equivalent to $\plim \hat{\tilde{B}} = \tilde{B}_0$, or $\Delta\tilde{B} \xrightarrow{p} \mathbf{0}$. For this to hold structurally, the term inside the projection must be zero. The condition $M_\Gamma \Cov[\tilde{A}_g \bm{\epsilon}_g^0, W_g] = \mathbf{0}$ simplifies to $U^\top \E[\tilde{A}_g \bm{\epsilon}_g^0 W_g^\top] = \mathbf{0}$.
This implies the bias term is zero, ensuring consistency. Conversely, if the covariance is non-zero, the unconstrained estimator is inconsistent, and consistent recovery of $B_0$ is generally impossible. Thus, the stated condition is necessary and sufficient.
\end{proof}

\subsection{Proof of Proposition \ref{prp:md_bias_bound}}

\begin{prop}[Bias Bound]
Let $(\hat {B}_1^{MD}, \hat {\bm{\alpha}}_1^{MD}, \{\hat \lambda^{MD}_{1,g}\}) $ be the solution using $\bm{\theta}_g$ as the outcome in the weighted regression \eqref{eq:feasib_md_weighted}. Define the estimation errors $\Delta B = \hat {B}_1^{MD} - B^{\star}$ and $\Delta \alpha = \hat {\bm{\alpha}}_1^{MD} - \bm{\alpha}^{\star}$. Let $M := \frac{1}{G}\sum_{g=1}^G \omega_g (1,W_{g})^\top (1,W_g)$. Then,
\begin{equation*}
 \left\| (\Delta \alpha, \Delta B) \right\|_F \le \frac{1}{\min(1, \kappa)} \left(\frac{\sqrt{1+\max_{g} \|W_g\|_2^2}}{\lambda_{\min}(M)}  \max_g \|\bm{r}^\star_g\|_2 \right) \frac{\|\omega - 1\|_1}{G}
\end{equation*}
where $\bm{r}^\star_g$ is the oracle residual, $\|\omega - 1\|_1 = \sum_{g=1}^G (1-\omega_g)$ is the number of discarded groups, and $\kappa$ is the identification constant from Assumption~\ref{as:identification_subspace}.
\end{prop}

\begin{proof}
Let $\bm{\theta}_g \in \mathbb{R}^k$ be the group-specific parameter vector. Let $\bm{\alpha} \in \mathbb{R}^k$ be the intercept and $B \in \mathcal{B}_0 \subset \mathbb{R}^{k \times p}$ be the matrix of policy effects. The objective function for the weighted regression is:
\begin{align*}
    L(\bm{\alpha}, B, \{\bm{\lambda}_g\}; \{\omega_g\}) = \sum_{g=1}^G \|\bm{\theta}_g - \bm{\alpha} - \Gamma \bm{\lambda}_g - B W_g\|_2^2 \omega_g
\end{align*}
For fixed $(\bm{\alpha}, B)$, minimizing with respect to the nuisance parameters $\bm{\lambda}_g$ yields the orthogonal projection of the residuals onto the complement of $\Gamma$. Let $P_{\Gamma^\perp} = I_k - \Gamma(\Gamma^\top \Gamma)^{-1}\Gamma^\top$. The concentrated objective function becomes:
\begin{align*}
    L'(\bm{\alpha}, B; \{\omega_g\}) = \sum_{g=1}^G \| P_{\Gamma^\perp}(\bm{\theta}_g - \bm{\alpha} - B W_g) \|_2^2 \omega_g
\end{align*}
Since $\Gamma^\top \bm{\alpha} = \mathbf{0}$, we have $P_{\Gamma^\perp} \bm{\alpha} = \bm{\alpha}$. For the policy effects, we define the \textit{projected} matrix $\tilde{B} := P_{\Gamma^\perp} B$. By Assumption~\ref{as:identification_subspace}, the mapping between $B \in \mathcal{B}_0$ and its projection $\tilde{B}$ is a bijection onto its image. We can therefore re-parameterize the objective in terms of the identifiable parameters $\tilde{X} = (\bm{\alpha}, \tilde{B})$. Let $\bm{\theta}'_g := P_{\Gamma^\perp} \bm{\theta}_g$. The objective simplifies to:
\begin{align*}
    \tilde{L}(\tilde{X}; \{\omega_g\}) = \sum_{g=1}^G \| \bm{\theta}'_g - \bm{\alpha} - \tilde{B} W_g \|_2^2 \omega_g
\end{align*}
The parameters $\tilde{X}$ operate in the subspace $\im(P_{\Gamma^\perp})$ of dimension $k'$. By choosing an orthonormal basis for this subspace, we can view $\tilde{X}$ as a $k' \times (1+p)$ matrix of coefficients.

Let $\tilde{X}_1 := (\hat{\bm{\alpha}}_1^{MD}, \tilde{B}_1^{MD})$ be the minimizer of $\tilde{L}(\tilde{X}; \{\omega_g\})$, and $\tilde{X}_0 := (\bm{\alpha}^\star, \tilde{B}^\star)$ be the minimizer of the unweighted objective ($\omega_g=1$). Define the estimation error in the projected space as $\Delta \tilde{X} = \tilde{X}_1 - \tilde{X}_0$. The oracle residual is $\bm{r}_g^\star$, and its projection is $\bm{r}'_g := P_{\Gamma^\perp}\bm{r}_g^\star = \bm{\theta}'_g - \bm{\alpha}^\star - \tilde{B}^\star W_g$.

The first-order condition for the weighted estimator is $\nabla \tilde{L}_1(\tilde{X}_1) = \mathbf{0}$. A Taylor expansion around $\tilde{X}_0$ yields:
\begin{align*}
    \nabla \tilde{L}_1(\tilde{X}_1) = \nabla \tilde{L}_1(\tilde{X}_0) + \tilde{H}_1 \Delta \tilde{X} = \mathbf{0} \implies \tilde{H}_1 \Delta \tilde{X} = - \nabla \tilde{L}_1(\tilde{X}_0)
\end{align*}
where $\tilde{H}_1$ is the Hessian. We analyze the gradient and Hessian separately.

\textbf{Gradient Bound:} The gradient $\nabla \tilde{L}_1(\tilde{X}_0)$ relates to the dropped groups. Since the unweighted gradient is zero ($\sum \nabla \tilde{L}_0^{(g)} = \mathbf{0}$), we have:
\begin{align*}
    -\nabla \tilde{L}_1(\tilde{X}_0) = - \sum_{g=1}^G (\omega_g - 1 + 1) \nabla \tilde{L}_0^{(g)}(\tilde{X}_0) = \sum_{g=1}^G (1 - \omega_g) \nabla \tilde{L}_0^{(g)}(\tilde{X}_0)
\end{align*}
The term $\nabla \tilde{L}_0^{(g)}(\tilde{X}_0)$ represents the gradient of the squared residual for group $g$ with respect to $\tilde{X}$. This is given by $-2(\bm{r}'_g, \bm{r}'_g W_g^\top)$. Let $G'_0 = 2 \sum_{g=1}^G (1 - \omega_g) (\bm{r}'_g, \bm{r}'_g W_g^\top)$. We bound its Frobenius norm:
\begin{align*}
\|G'_0\|_F &= \left\| 2 \sum_{g=1}^G (1 - \omega_g) (\bm{r}'_g, \bm{r}'_g W_g^\top) \right\|_F \\
&\le 2 \sum_{g=1}^G |1 - \omega_g| \|\bm{r}'_g\|_2 \sqrt{1 + \|W_g\|_2^2} \\
&\le 2 \left( \max_g \|\bm{r}_g^\star\|_2 \right) \sqrt{1 + \max_g \|W_g\|_2^2} \|\omega - 1\|_1
\end{align*}
where we used $\|\bm{r}'_g\|_2 = \|P_{\Gamma^\perp} \bm{r}_g^\star\|_2 \le \|\bm{r}_g^\star\|_2$ (projection is non-expansive).

\textbf{Hessian Bound:} The Hessian $\tilde{H}_1$ corresponds to the second derivatives of $\tilde{L}$. Because $\bm{\alpha}$ and $\tilde{B}$ act linearly and the quadratic loss is separable across the orthogonal dimensions of the projected space, the problem decomposes into independent weighted least squares problems sharing the same design matrix. The Hessian for each scalar component is $2 \sum \omega_g (1, W_g)^\top (1, W_g) = 2G M$.
Thus, the smallest singular value of the full Hessian $\tilde{H}_1$ acting on $\Delta \tilde{X}$ is $2G \lambda_{\min}(M)$. This implies:
\begin{align*}
    \|\Delta \tilde{X}\|_F \le \frac{1}{2G \lambda_{\min}(M)} \|G'_0\|_F
\end{align*}
Substituting the gradient bound:
\begin{align*}
\|\Delta \tilde{X}\|_F \le \left(\frac{\sqrt{1+\max_{g} \|W_g\|_2^2}}{\lambda_{\min}(M)}  \max_g \|\bm{r}^\star_g\|_2 \right) \frac{\|\omega - 1\|_1}{G}.
\end{align*}
\textbf{Structural Parameter Bound:} We must recover the bound for the structural parameters $\Delta X = (\Delta \bm{\alpha}, \Delta B)$. Note that $\|\Delta X\|_F^2 = \|\Delta \bm{\alpha}\|_F^2 + \|\Delta B\|_F^2$.
Since $\bm{\alpha}$ lies in the image of $P_{\Gamma^\perp}$, we have $\|\Delta \bm{\alpha}\|_F = \|P_{\Gamma^\perp} \Delta \bm{\alpha}\|_F = \|\Delta \tilde{\bm{\alpha}}\|_F$.
For the policy effects, Assumption~\ref{as:identification_subspace} implies $\|P_{\Gamma^\perp} \Delta B\|_F \ge \kappa \|\Delta B\|_F$. Since $\Delta \tilde{B} = P_{\Gamma^\perp} \Delta B$, we have $\|\Delta B\|_F \le \frac{1}{\kappa} \|\Delta \tilde{B}\|_F$.
Combining these:
\begin{align*}
    \|\Delta X\|_F^2 &\le \|\Delta \tilde{\bm{\alpha}}\|_F^2 + \frac{1}{\kappa^2} \|\Delta \tilde{B}\|_F^2 \\
    &\le \frac{1}{\min(1, \kappa^2)} \left( \|\Delta \tilde{\bm{\alpha}}\|_F^2 + \|\Delta \tilde{B}\|_F^2 \right) \\
    &= \frac{1}{\min(1, \kappa)^2} \|\Delta \tilde{X}\|_F^2
\end{align*}
Taking the square root and using the bound for $\|\Delta \tilde{X}\|_F$ completes the proof.
\end{proof}

\subsection{Proof of Proposition \ref{prop:md_inference}}
\begin{proof}
The feasible MD estimator $\hat{B}^{MD}$ can be decomposed based on its relationship with the oracle estimator $B^\star$. We write the total estimation error as:
\begin{align*}
    \hat{B}^{MD} - B_0 = (\hat{B}_1^{MD} - B^\star) + \hat{B}_0^{MD} + (B^\star - B_0).
\end{align*}
The third term, $\sqrt{G}(B^\star - B_0)$, converges to the limiting normal distribution $N(0, V_{B^\star})$ by assumption. We need to show that the first two components, when scaled by $\sqrt{G}$, vanish in probability:
\begin{enumerate}
    \item[(A)] $\sqrt{G}(\hat{B}_1^{MD} - B^\star) \xrightarrow{p} \mathbf{0}$
    \item[(B)] $\sqrt{G}\hat{B}_0^{MD} \xrightarrow{p} \mathbf{0}$
\end{enumerate}

\paragraph{Part (A):} Proposition~\ref{prp:md_bias_bound} provides a bound for the Frobenius norm of the identifiable parameters. Let $\Delta \tilde{X} = (\Delta \bm{\alpha}, \Delta \tilde{B})$ be the estimation error in the projected space. The proposition establishes:
\begin{align*}
\normF{\Delta \tilde{X}} \le K \frac{\sum_{g=1}^G (1-\omega_g)}{G}
\end{align*}
where $K = \frac{\sqrt{1+M_W^2}}{\lambda_{\min}(M)}(\max_g \normtwo{\bm{r}_g^\star})$. Under standard regularity conditions (bounded regressors, bounded residuals, and positive definite design matrix), $K = O_p(1)$.

Let $N_{excl} = \sum_{g=1}^G (1-\omega_g)$ be the number of excluded groups. The expectation is $\E[N_{excl}] = \sum_{g=1}^G \mathbb{P}[\omega_g=0]$. By assumption, $\sum_{g=1}^G \mathbb{P}[\omega_g=0] \ll \sqrt{G}$, which implies $\E[N_{excl}] = o(\sqrt{G})$. By Markov's inequality, for any $\epsilon > 0$, $P(N_{excl} \ge \epsilon\sqrt{G}) \le \E[N_{excl}] / (\epsilon\sqrt{G})$. Since the RHS converges to 0, it follows that $N_{excl} = o_p(\sqrt{G})$.
Therefore,
\begin{align*}
\normF{\Delta \tilde{X}} = O_p(1) \cdot \frac{o_p(\sqrt{G})}{G} = o_p\left(\frac{1}{\sqrt{G}}\right).
\end{align*}
This implies $\sqrt{G} \normF{\Delta \tilde{X}} = o_p(1)$.
By Assumption~\ref{as:identification_subspace} $\|\hat{B}_1^{MD} - B^\star\|_F \sim \|\Delta \tilde{B}\|_F \le \|\Delta \tilde{X}\|_F$. Consequently, $\sqrt{G}(\hat{B}_1^{MD} - B^\star) \xrightarrow{p} \mathbf{0}$.

\paragraph{Part (B):}
Assume $\frac{1}{G}\sum_{g=1}^G W_g \xrightarrow{p} \mu_W$ and $\frac{1}{G}\sum_{g=1}^G W_g W_g^\top \xrightarrow{p} \Sigma_{WW}$ for deterministic limits, with $\Sigma_{WW} - \mu_{W}\mu_{W}^\top > 0$.
The component $\hat{B}_0^{MD}$ corresponds to the estimator solving the weighted MD problem where the outcome variable is the first-stage estimation error $\bm{\varepsilon}_g = \hat{\bm{\theta}}_g - \bm{\theta}_g$. The target parameters for this regression are $\mathbf{0}$. We analyze the estimator in the projected space, $\hat{\tilde{B}}_0^{MD}$, which solves the weighted normal equations:
\begin{align*}
\left(\sum_{g=1}^G \omega_g I_{k'}\right) \hat{\tilde{\bm{\alpha}}}_0^{MD} + \left(\sum_{g=1}^G \omega_g I_{k'} \hat{\tilde{B}}_0^{MD} W_g\right) &= \sum_{g=1}^G \omega_g U^\top \bm{\varepsilon}_g \\
\left(\sum_{g=1}^G \omega_g (W_g \otimes I_{k'})\right) \hat{\tilde{\bm{\alpha}}}_0^{MD} + \left(\sum_{g=1}^G \omega_g (W_g W_g^\top \otimes I_{k'})\right) \vectorize(\hat{\tilde{B}}_0^{MD}) &= \vectorize\left(\sum_{g=1}^G \omega_g U^\top \bm{\varepsilon}_g W_g^\top\right).
\end{align*}
Let $\tilde{H}_{11}^\omega, \tilde{H}_{12}^\omega, \tilde{H}_{21}^\omega, \tilde{H}_{22}^\omega$ denote the blocks of the Hessian as defined in the previous proof. Let $C_1^{\varepsilon,\omega}$ and $C_2^{\varepsilon,\omega}$ denote the score terms on the RHS.

First, we establish the probability limits of the Hessian. Since $\E[N_{excl}] = o(\sqrt{G})$, we have $\frac{1}{G}\sum_{g=1}^G (1-\omega_g) \xrightarrow{p} 0$, so $\frac{1}{G}\sum \omega_g \xrightarrow{p} 1$.
It follows that $\frac{1}{G}\tilde{H}_{11}^\omega \xrightarrow{p} I_{k'} =: H_{11,plim}$.
Similarly, for the cross term $\tilde{H}_{21}^\omega$:
\begin{align*}
    \normF{\frac{1}{G}\sum_{g=1}^G \omega_g (W_g \otimes I_{k'}) - \frac{1}{G}\sum_{g=1}^G (W_g \otimes I_{k'})} &\le \sqrt{k'}C_W \left(\frac{1}{G}\sum_{g=1}^G (1-\omega_g)\right) \xrightarrow{p} 0.
\end{align*}
Thus, $\frac{1}{G}\tilde{H}_{21}^\omega \xrightarrow{p} \mu_W \otimes I_{k'}$. Likewise, $\frac{1}{G}\tilde{H}_{22}^\omega \xrightarrow{p} \Sigma_{WW} \otimes I_{k'}$.
The Schur complement $\tilde{S}^\omega$ converges in probability to $S_{plim} = (\Sigma_{WW} - \mu_W\mu_W^\top) \otimes I_{k'}$, which is positive definite by assumption.

Next, we examine the scaled score terms $\frac{1}{\sqrt{G}}C_1^{\varepsilon,\omega}$ and $\frac{1}{\sqrt{G}}C_2^{\varepsilon,\omega}$.
Define $\bm{\xi}_{1g} := \omega_g U^\top \bm{\varepsilon}_g$. This has mean zero since $\E[\bm{\varepsilon}_g | \omega_g=1] = \mathbf{0}$. The variance is:
\begin{align*}
    \Var\left[\frac{1}{\sqrt{G}}C_1^{\varepsilon,\omega}\right] &= \frac{1}{G}\sum_{g=1}^G \E[\omega_g U^\top \bm{\varepsilon}_g \bm{\varepsilon}_g^\top U] = \frac{1}{G}\sum_{g=1}^G \E\left[\omega_g U^\top \left(\frac{1}{n_g}\check{\Sigma}_{0,g}\right) U \right] \\
    &\le \frac{C_{\Sigma}}{G} \sum_{g=1}^G \frac{\mathbb{P}(\omega_g=1)}{n_g} \le C_{\Sigma} \left( \frac{1}{G} \sum_{g=1}^G \frac{1}{n_g} \right).
\end{align*}
By Lemma~\ref{lem:ng_growth_implication}, the condition $\sum_{g=1}^G e^{-cn_g} \ll \sqrt{G}$ implies that $\frac{1}{G}\sum_{g=1}^G \frac{1}{n_g} \to 0$. Thus, $\Var[\frac{1}{\sqrt{G}}C_1^{\varepsilon,\omega}] \to \mathbf{0}$.
By Chebyshev's inequality, $\frac{1}{\sqrt{G}}C_1^{\varepsilon,\omega} \xrightarrow{p} \mathbf{0}$.
A similar logic applies to $C_2^{\varepsilon,\omega}$, where the variance is bounded by a constant times $\frac{1}{G}\sum \frac{1}{n_g}$ due to the boundedness of $W_g$. Thus, $\frac{1}{\sqrt{G}}C_2^{\varepsilon,\omega} \xrightarrow{p} \mathbf{0}$.

Substituting these limits into the solution for $\hat{\tilde{B}}_0^{MD}$:
\begin{align*}
    \sqrt{G}\vectorize(\hat{\tilde{B}}_0^{MD}) &= \left(\frac{1}{G}\tilde{S}^\omega\right)^{-1} \left( \frac{1}{\sqrt{G}}C_2^{\varepsilon,\omega} - \left(\frac{1}{G}\tilde{H}_{21}^\omega\right) \left(\frac{1}{G}\tilde{H}_{11}^\omega\right)^{-1} \frac{1}{\sqrt{G}}C_1^{\varepsilon,\omega} \right) \\
    &\xrightarrow{p} S_{plim}^{-1} \times \mathbf{0} = \mathbf{0}.
\end{align*}
This implies $\sqrt{G}\hat{\tilde{B}}_0^{MD} \xrightarrow{p} \mathbf{0}$. Finally, by Assumption~\ref{as:identification_subspace}, $\sqrt{G}\hat{B}_0^{MD} \xrightarrow{p} \mathbf{0}$.
\end{proof}

\subsection{Technical lemmas}
\begin{lemma} \label{lem:ng_growth_implication}
Let $\{n_g\}_{g=1}^G$ be a sequence of positive integers such that $n_g \ge 1$ for all $g$. Let $c$ be a strictly positive constant ($c > 0$). If the condition
\begin{align*}
\lim_{G\to\infty} \frac{\sum_{g=1}^G e^{-cn_g}}{\sqrt{G}} = 0
\end{align*}
holds, then it implies that
\begin{align*}
\lim_{G\to\infty} \frac{1}{G}\sum_{g=1}^G \frac{1}{n_g} = 0.
\end{align*}
\end{lemma}

\begin{proof}
Assume that the implication does not hold. This means that there exists a $\delta > 0$ and a subsequence $G_k \to \infty$ such that for all $k$:
\begin{align} \label{eq:lemma_contradiction_assumption}
\frac{1}{G_k}\sum_{g=1}^{G_k} \frac{1}{n_g} \ge \delta
\end{align}
Since $n_g \ge 1$, we have $1/n_g \le 1$. Thus, $0 < \delta \le 1$.
Let $M_0$ be a constant chosen such that $M_0 > 1/\delta$.

For each $G_k$ in the subsequence, partition the indices $\{1, \dots, G_k\}$ into two sets:
$S_k = \{g \in \{1, \dots, G_k\} \mid n_g \le M_0\}$
$L_k = \{g \in \{1, \dots, G_k\} \mid n_g > M_0\}$
Let $|S_k|$ denote the number of elements in $S_k$. So, $|L_k| = G_k - |S_k|$. From \eqref{eq:lemma_contradiction_assumption}, we have:
\begin{align*}
\delta &\le \frac{1}{G_k} \left( \sum_{g \in S_k} \frac{1}{n_g} + \sum_{g \in L_k} \frac{1}{n_g} \right) 
\le \frac{|S_k|}{G_k} \left(1 - \frac{1}{M_0}\right) + \frac{1}{M_0}
\end{align*}
Let $f_k = |S_k|/G_k$ be the fraction of groups for which $n_g \le M_0$. Rearranging the inequality:
\begin{align*}
f_k \left(\frac{M_0-1}{M_0}\right) &\ge \delta - \frac{1}{M_0} = \frac{\delta M_0 - 1}{M_0}
\end{align*}
Since we chose $M_0 > 1/\delta$, $\delta M_0 - 1 > 0$. Also, since $M_0 > 1$, $M_0-1 > 0$. Let $\eta = \frac{\delta M_0 - 1}{M_0 - 1} >0$. For all $k$, $f_k = |S_k|/G_k \ge \eta$, which means $|S_k| \ge \eta G_k$, and thus a fixed positive fraction $\eta$ of the groups have their $n_g$ bounded by $M_0$.

Now consider the sum $\sum_{g=1}^{G_k} e^{-cn_g}$:
\begin{align*}
\sum_{g=1}^{G_k} e^{-cn_g} \ge \sum_{g \in S_k} e^{-cn_g} \ge |S_k| e^{-cM_0} \ge \eta G_k e^{-cM_0} \Rightarrow \frac{\sum_{g=1}^{G_k} e^{-cn_g}}{\sqrt{G_k}} \ge \eta \sqrt{G_k} e^{-cM_0}
\end{align*}
As $k \to \infty$, $G_k \to \infty$, so $\sqrt{G_k} \to \infty$. 
Therefore,
\begin{align*}
\lim_{k\to\infty} \eta \sqrt{G_k} e^{-cM_0} = \infty
\end{align*}
This implies that $\liminf_{k\to\infty} \frac{\sum_{g=1}^{G_k} e^{-cn_g}}{\sqrt{G_k}} = \infty$ leading to a contradiction. 
\end{proof}

\clearpage

\section{Additional examples}
    \label{app:examples}
    
\subsection{Composition Effects}
\label{sec:aggregate_interpretation}

As foreshadowed by Remark~\ref{rem:did_heterogeneity}, aggregate parameters may reflect not only direct policy effects but also policy-induced shifts in the composition of heterogeneous micro-units. This section uses the child penalty example to elaborate on this compositional challenge, showing how it complicates causal inference.

Consider again the DiD setup for estimating the child penalty (Section \ref{subsec:motivating_example}), but now allow for explicit unit-level heterogeneity in the penalty itself. Suppose the outcome $Y_{g,i,t}$ follows a linear model:
\begin{align*}
 Y_{g,i,t} = \gamma_{g,i} + \delta_{g,t} + \tau_{g,i} E_{g,i} \mathbf{1}\{t=2\} + \epsilon_{g,i,t}, \quad \mathbb{E}_{F_g}[\epsilon_{g,i,t} | E_{g,i}] = 0,
\end{align*}
where $E_{g,i}$ indicates the first birth event between $t=1$ and $t=2$, and $\tau_{g,i}$ is the individual-specific child penalty for unit $i$ in group $g$. A common aggregate parameter of interest in group $g$ is the average penalty among those who experience the event (i.e., the Average Treatment Effect on the Treated, ATT):
\begin{align*}
 \tau_g = \mathbb{E}_{F_g}[\tau_{g,i} | E_{g,i}=1].
\end{align*}
This parameter represents the average child penalty for mothers in group $g$ (e.g., municipality) and is identified by population moments discussed in Section \ref{subsec:motivating_example}.

Consider a two-dimensional group-level policy $W_g = (W_{g,1}, W_{g,2})^\top$, with randomly assigned, yet potentially correlated, components. Suppose $W_{g,1}$ (e.g., local labor market conditions or social norms) influences selection into parenthood but not the individual child penalty $\tau_{g,i}$. Conversely, suppose $W_{g,2}$ (e.g., generosity of family policies like paid leave) directly affects $\tau_{g,i}$ but not selection into parenthood.

Under these conditions, the potential outcome for the aggregate ATT parameter $\tau_g(w)$ for a given policy vector $w=(w_1, w_2)$ is derived from its definition:
\begin{align*}
 \tau_g(w) = \mathbb{E}_{F_g(w)}[\tau_{g,i}(w_2) | E_{g,i}(w_1)=1].
\end{align*}
This expectation averages the micro-level penalties $\tau_{g,i}(w_2)$ over the subpopulation of parents, defined by the condition $E_{g,i}(w_1)=1$. Crucially, because $W_{g,1}$ affects selection into motherhood via its influence on $F_g(w)$, the aggregate parameter $\tau_g(w)$ generally depends on $W_{g,1}$ solely through this compositional channel, even though the underlying micro-level penalties $\tau_{g,i}$ only respond to $W_{g,2}$. This implies that the structural relationship for the aggregate parameter is more appropriately modeled as $\tau_g = \alpha + \beta_1 W_{g,1} + \beta_2 W_{g,2}$, where $\beta_1$ captures the composition effect and $\beta_2$ captures the direct effect of family policy generosity on the average penalty. 

This dependence structure has critical implications for estimation strategies. First, consider approaches that misspecify the model by omitting $W_{g,1}$. A researcher, correctly believing that the individual penalty $\tau_{g,i}$ only depends causally on the family policy $W_{g,2}$, might estimate a single equation using OLS:
\begin{align*}
\min_{\{\gamma_{g,i}\}, \{\delta_{g,t}\}, \alpha, \beta_{2}} \sum_{g,i,t} \left( Y_{g,i,t} - \gamma_{g,i} + \delta_{g,t} - (\alpha + \beta_{2} W_{g,2}) E_{g,i} \mathbf{1}\{t=2\} \right)^2.
\end{align*}
This approach generically produces an inconsistent estimate $\hat{\beta}_{2}$, suffering from standard omitted variable bias (since $W_{g,1}$ influences $\tau_g$ through composition and is likely correlated with $W_{g,2}$). Similarly, implementing a two-stage MD approach by first estimating $\hat{\tau}_g$ within each group and then running the second-stage regression,
\begin{align*}
 \min_{\alpha, \beta_2} \sum_{g} (\hat{\tau}_g - \alpha - \beta_2 W_{g,2})^2,
\end{align*}
also yields an inconsistent estimate $\hat{\beta}_2$ due to the same omitted variable issue. This occurs despite the fact that the individual-level penalty $\tau_{g,i}$ causally depends only on $W_{g,2}$.

Now consider approaches that correctly specify the model for the aggregate $\tau_g$ by including both policy components. One might be tempted to estimate the appropriately specified equation using OLS:
\[
\min_{\{\gamma_{g,i}\}, \{\delta_{g,t}\}, \alpha, \beta_{1}, \beta_{2}} \sum_{g,i,t} \left( Y_{g,i,t} -  \gamma_{g,i} - \delta_{g,t} - (\alpha + \beta_{1} W_{g,1} + \beta_{2} W_{g,2}) E_{g,i} \mathbf{1}\{t=2\} \right)^2.
\]
However, as we formally demonstrate in Section \ref{sec:identification}, this one-step GMM/OLS approach still generally yields inconsistent estimates $\hat{\beta}_1$ and $\hat{\beta}_2$. The reason is the endogenous weighting bias: the implicit GMM weights depend on $W_{g,1}$ (via its effect on the distribution of $E_{g,i}$, which determines the sample moments used in the equivalent GMM formulation), violating the conditions required for consistency.

Finally, consider the two-stage MD approach using the correctly specified second stage: first estimate $\hat{\tau}_g$ (the average penalty for mothers in group $g$) in the first stage, then run the second-stage regression:
\[
\min_{\alpha, \beta_1, \beta_2}\sum_{g}\left(\hat{\tau}_g - \alpha - \beta_1 W_{g,1} - \beta_2 W_{g,2}\right)^2.
\]
Under certain regularity conditions (including sufficiently large $n_g$, as discussed in Section \ref{sec:md_estimator}), this approach yields consistent estimates $\hat{\beta}_1$ and $\hat{\beta}_2$. The first stage consistently estimates the parameter $\tau_g(W_{g,1}, W_{g,2})$ as defined by the population moments. The second stage then correctly relates this parameter to the exogenously assigned policy variables $W_{g,1}$ and $W_{g,2}$. This avoids both the omitted variable bias (present in models omitting $W_{g,1}$) and the endogenous weighting bias (present in the OLS/GMM).

\subsection{IV analysis} \label{sec:extensions}

This section discusses an extension of our analysis to models where a policy variable of interest is endogenous, necessitating an instrumental variable (IV) approach. To illustrate concretely the endogenous weighting issues analogous to those discussed for GMM estimators in Section~\ref{sec:identification}, we focus our discussion on the estimation of how local policies affect returns to schooling. 

We start by augmenting the framework from Section~\ref{sec:framework}. For each group $g$ (e.g., a local labor market), assume the data-generating process for individual outcomes, $F_g$, can depend on the group-level instrument $Z_g$, i.e., $F_g = F_g(Z_g)$. A local policy of interest, $W_g$, is a group-specific function of $Z_g$, i.e., $W_g = W_g(Z_g)$. We assume that the triplet $(Z_g, W_g(\cdot), F_g(\cdot))$, is i.i.d. across groups and $Z_g$ is independent of $W_g(\cdot)$ and $F_g(\cdot)$.

Consider individuals $i$ in market $g$. We observe log wages $Y_{g,i}$ and years of schooling $S_{g,i}$. The local returns to schooling, $\tau_g$, are determined by a Mincer-style equation:
\begin{equation}\label{eq:iv_mincer_model}
Y_{g,i} = \delta_g + \tau_g S_{g,i} + u_{g,i}, \quad \E_{F_g}[u_{g,i} | Q_{g,i}] = 0,
\end{equation}
where potentially endogenous schooling $S_{g,i}$ (e.g., due to ability bias) is instrumented via $Q_{g,i}$ (e.g., quarter of birth). The data for group $g$ are $\{(Y_{g,i}, S_{g,i}, Q_{g,i})\}_{i=1}^{n_g}$, drawn i.i.d. from $F_g$. For simplicity, we assume $n_g=n$. Model \eqref{eq:iv_mincer_model} fits within our general setup, with $\bm{\theta}_g := (\delta_g, \tau_g)^\top$.

The local returns to schooling $\tau_g$ are modeled as a function of an endogenous local policy $W_g$:
\begin{equation*}
\tau_g =  \alpha_g + \beta W_g,
\end{equation*}
where $\beta$ is the structural effect of policy $W_g$ on returns, and $\alpha_g$ captures unobserved market-level heterogeneity in returns, potentially correlated with $W_g$. We employ a group-level instrument $Z_g$ for the policy $W_g$, which by assumption satisfies
\begin{align*}
    \E\left[(Z_g - \mu_Z)\begin{pmatrix} 1 \\ \alpha_g \end{pmatrix}\right] = \bm{0}_2,
\end{align*}
where $\mu_Z := \E[Z_g]$. A key premise is that the policy instrument $Z_g$ (via $W_g$) may affect not only $\tau_g$ but also the distribution of schooling $S_{g,i}$ or its relationship with the instrument $Q_{g,i}$ within market $g$. In contrast, we assume the conditional distribution of $Q_{g,i}$ (given group $g$ and individual characteristics other than $S_{g,i}$) does not directly depend on $Z_g$.

A common empirical strategy is then to estimate the equation:
\begin{equation*}
Y_{g,i} = \delta_g + (\theta_0 + \beta W_g) S_{g,i} + \epsilon_{g,i},
\end{equation*}
via TSLS.  The terms $S_{g,i}$ and $W_g S_{g,i}$ are treated as endogenous and instrumented using $Q_{g,i}$ and $Z_g Q_{g,i}$, respectively. Let $\beta^{TSLS}$ denote the probability limit of this TSLS estimator for $\beta$ as $G \rightarrow \infty$ and $n \rightarrow \infty$. One can show that it has the following representation:
\begin{equation*}
    \beta^{TSLS} = \beta + \frac{\mathbb{C}\text{ov}^{C_g(Z_g)}[\alpha_g, Z_g]}{\mathbb{C}\text{ov}^{C_g(Z_g)}[W_g, Z_g]},
\end{equation*}
where the weighted covariance relies on weights
\begin{align*}
    C_g := \mathbb{C}\text{ov}_{F_g}[S_{g,i}, Q_{g,i}],
\end{align*}
which generically depend on $Z_g$, and thus we write $C_g = C_g(Z_g)$.

The TSLS estimator $\beta^{TSLS}$ is thus generically inconsistent for $\beta$. The bias arises if the $C_g$-weighted covariance between $\alpha_g$ (unobserved heterogeneity in returns) and $Z_g$ is non-zero. This can occur even though the unweighted covariance, $\mathbb{C}\text{ov}[\alpha_g, Z_g]$, is zero by assumption, provided $C_g$ is itself a function of $Z_g$. Such dependence of $C_g$ on $Z_g$ emerges if the policy instrument $Z_g$ (via $W_g$) affects the joint distribution of $S_{g,i}$ and $Q_{g,i}$ (e.g., by altering schooling levels or the composition of individuals for whom $Q_{g,i}$ strongly predicts $S_{g,i}$). This introduces an endogenous weighting bias, analogous to the one discussed in Section~\ref{sec:identification}.

These findings may seem at odds with design-based identification results for shift-share IV estimators \citep{adao2019shift, borusyak2022quasi} or, more generally, for "formula IV" estimators \citep{Borusyak2023, Borusyak2024a}, which show consistency of appropriately recentered IV estimators.  The apparent contradiction is resolved by recognizing that our probability model allows the distribution of $S_{g,i}$ to vary with $Z_g$, whereas conventional design-based analyses often treat such characteristics or their distributions as fixed with respect to instrument assignment. If $S_{g,i}$'s distribution indeed covaries with $Z_g$, comprehensive recentering of $Z_g$ requires conditioning on $S_{g,i}$. Standard correction by $\mu_Z$ is insufficient to eliminate the bias, as our analysis demonstrates.

\begin{remark}
    The discussion above focuses on GMM estimation. The MD estimation is standard in this case with one caveat: moments induced by \eqref{eq:iv_mincer_model} do not satisfy the restrictions in Section \ref{sec:md_estimator} and thus unbiased estimation of $\tau_g$ is generally impossible with finite $n_g$. This puts additional restrictions on the magnitude of $n_g$ and may require using bias correction methods, analogous to those discussed in Remark \ref{rem:uncondition_moments}.
\end{remark}

\subsection{Additional examples}

\subsubsection{Job Displacement, Average Wage Premia, Policy Interactions, and Sorting}\label{sec:akm_example}

Our framework can be used to analyze the effects of group-level policies on local labor market structures, drawing inspiration from models in the style of \citet{Abowd1999}. We consider a setting where policies implemented at the level of a local labor market (LLM), indexed by $g$, might influence not only average wage parameters but also patterns of worker sorting.

Consider log-wages $Y_{g,i,t}$ for worker $i$ in LLM $g$ at time $t \in \{1, 2\}$. Workers are employed by firms classified into $J=2$ types (e.g., small and large), denoted by $j(g,i,t) \in \{1, 2\}$. A standard specification decomposes wages as:
\begin{align*}
    Y_{g,i,t} = \gamma_{g,i} + \psi_{g,j(g,i,t)} + \epsilon_{g,i,t}, \quad \E_{F_{g}}[ \epsilon_{g,i,t}| \{j(g,i,t')\}_{t'=1,2}] = 0.
\end{align*}
where $\gamma_{g,i}$ is a worker-specific effect, and $\psi_{g,j}$ represents the wage premium associated with firm type $j$ in market $g$. For identification purposes we normalize $\psi_{g,1}=0$ for each $g$.

Consider a researcher interested in how a job displacement event alters the firm premia effectively experienced by workers, and whether this alteration varies with policy $W_g$ . Let $T_{g,i}$ be an indicator for worker $i$ in group $g$ being affected by a job displacement occurring between periods $t=1$ and $t=2$. The research question centers on how this layoff event, interacts with the firm premium structure in the post-layoff period ($t=2$), and how policy $W_g$ changes this interaction. We can formalize this analysis using the following DiD-type specification:
\begin{align*}
\psi_{g,j(i,t)} =\mu_{g,i} + \delta_{g,t} + \tau_g T_{g,i}\mathbf{1}\{t=2\} + \nu_{g,i,t}, \quad \mathbb{E}_{F_g}[\nu_{g,i,t}|T_{g,i}] = 0.
\end{align*}
Using the same notation as in the previous section we define $\bm{\theta}_g := (\tilde\delta_g, \tau_g)^\top$. The goal is to relate $\bm{\theta}_g$, in particular the effect of the job displacement event $\tau_g$, to a policy of interest $W_g$. For instance, in \cite{Daruich2023}, the authors are interested in how workers' firm wage premia after a job-displacement event ($\tau_g$) depends on lifting constraints on the employment of temporary contract workers.\footnote{In other studies, the authors relate the same parameter to unemployment rates \citep{Schmieder2023} or variation in labor-market structures across countries \citep{Bertheau2023}}

In practice, such an investigation often proceeds in two stages. First, the firm-type premium $\hat{\psi}_{g,2}$ is estimated for each group $g$ from individuals moving between firm types. Subsequently, these group-specific premium estimates inform an individual-level outcome, $\hat{\psi}_{g,j(i,t)}$. This constructed outcome is then used in a panel regression to estimate the policy effect. Specifically, parameters $(\alpha, \beta)$ along with nuisance parameters $(\{\gamma_{g,i}\}, \{\delta_{g,t}\})$ are estimated via OLS:
\begin{align*}
(\hat{\alpha}, \hat{\beta}, \{\hat{\gamma}_{g,i}\}, \{\hat{\delta}_{g,t}\}) = \underset{\alpha, \beta, \{\mu_{g,i}\}, \{\delta_{g,t} \}}{\operatorname{argmin}} \sum_{g,i,t} (\hat{\psi}_{g,j(i,t)} - \mu_{g,i} - \delta_{g,t}  - (\alpha + \beta W_g)T_{g,i}\mathbf{1}\{t=2\})^2.
\end{align*}
The coefficient on the interaction term $T_{g,i}\mathbf{1}\{t=2\}$ effectively models the group-specific layoff impact as $\tau_g = \alpha + \beta W_g$, thereby capturing the policy-dependent differential effect. 

This empirical strategy fits within our moment-based framework. The first moment condition, identifying $\psi_{g,2}$ from the wage changes of movers ($\Delta Y_{g,i}$), is standard:
\begin{align*}
\mathbb{E}_{F_g}[(\Delta Y_{g,i} - \psi_{g,2})\mathbf{1}\{h(i)=(1,2)\} + (\Delta Y_{g,i} + \psi_{g,2})\mathbf{1}\{h(i)=(2,1)\}] = 0.
\end{align*}
The second set of moment conditions identifies the change in the group-time effect, $\Delta \delta_g := \delta_{g,2}-\delta_{g,1}$, and the parameter $\tau_g$. Differencing the population relationship for $\psi_{g,j(i,t)}$ with respect to time implies that the experienced change in an individual's firm premium, $\Delta\psi^{\text{exp}}_{g,i} := \psi_{g,j(i,2)} - \psi_{g,j(i,1)}$, satisfies the following moment restriction:
\begin{align*}
\mathbb{E}_{F_g}\left[\begin{pmatrix}
    1 \\ T_{g,i}
\end{pmatrix}(\Delta\psi^{\text{exp}}_{g,i} - \tilde\delta_g - \tau_g T_{g,i})\right] = \mathbf{0}_2.
\end{align*}

This two-stage procedure---initial estimation of $\psi_{g,2}$ followed by its use in a panel OLS regression---can be represented as a specific one-step GMM estimator \eqref{eq:gmm_est}. The corresponding GMM weighting matrix $A_g$ reflects this sequential estimation logic, featuring a block-diagonal structure. The block corresponding to the $(\tilde\delta_g, \tau_g)$ moments would be based on the sample covariance of the regressors.

We can further extend this example to discuss sorting and heterogeneity. Analyzing the full impact of the policy often requires looking beyond its effect on average premia $\psi_{g,k}$. For instance, does policy $W_{g,k}$ targeted at type $k$ firms disproportionately attract high-ability workers ($\gamma_{g,i}$) to that type, or does it alter the dispersion of worker abilities within type $k$?

Simplifying to two firm types ($j=1, 2$), define a worker's mobility history
\begin{align*}
    h(i) := (j(g,i,1), j(g,i,2)),
\end{align*}
with four possible histories $h \in H = \{(1,1), (2,2), (1,2), (2,1)\}$. Using previous moment conditions we can identify the mean worker effect conditional on history, $\mu_{g,\gamma,h} = \E_{F_g}[\gamma_{g,i} | h(i)=h]$, which reveals patterns of worker sorting across different mobility paths. To identify the dispersion patterns  we need stronger restrictions:
\begin{align}
\label{eq:akm_sorting_assumptions}
    \E_{F_{g}}[ \epsilon_{g,i,t}|\gamma_{g,i},h(i)] = 0, \quad 
    \E_{F_{g}}[\epsilon_{g,i,1}\epsilon_{g,i,2}|h(i)] = 0.
\end{align}
The first condition strengthens the conditional moment restriction necessary for identifying the wage premium by conditioning on the worker effect $\gamma_{g,i}$, while the second assumes idiosyncratic errors are serially uncorrelated conditional on the mobility path. Armed with these restrictions we can identify the variance of worker effects conditional on history, $\sigma^2_{g,\gamma,h} = \Var_{F_g}[\gamma_{g,i} | h(i)=h]$, which captures the degree of heterogeneity within specific worker-path groups. Estimating the causal effect of a policy $W_g$ on this expanded set of parameters allows for a more comprehensive analysis of the policy's impact on market structure. As detailed in Remark \ref{rem:akm_linear_moments} below, this full vector $\bm{\theta}_g$ (including $\psi_{g,2}$, the four means $\{\mu_{g,\gamma,h}\}_h$, and the four variances $\{\sigma^2_{g,\gamma,h}\}_h$) is identified via a system of linear moment conditions. However, the reliance on worker mobility for identification persists, and the potential for the policy $W_g$ to influence these mobility patterns remains a central challenge for estimation via either GMM or MD methods.

The extended analysis under assumptions \eqref{eq:akm_sorting_assumptions} fits within our general moment-based framework $\mathbb{E}_{F_g}[h(D_{g,i}, \bm{\theta}_g)] = \mathbf{0}_k$. This example highlights how identifying richer features of the economic structure often necessitates stronger assumptions on the underlying microdata process. Furthermore, it illustrates how the policy of interest $W_g$ might affect the very variation (worker mobility) needed for identification, motivating the subsequent detailed analysis of estimation procedures in Sections \ref{sec:identification} and \ref{sec:md_estimator}.

\begin{remark}[Linear Moment Conditions for Full Vector]
\label{rem:akm_linear_moments} 
The full 9-dimensional parameter vector $\bm{\theta}_g = (\psi_{g,2}, \{\mu_{g,\gamma,h}\}_h, \{\sigma^2_{g,\gamma,h}\}_h)$ for the $J=2$ case is identified via a system of 9 linear moment conditions consistent with the structure $\E_{F_g}[h_1(D_{g,i}) - h_2(D_{g,i})\bm{\theta}_g] = \mathbf{0}$. These arise from: (1) wage changes of movers $1 \to 2$ (identifying $\psi_{g,2}$); (2) average period-1 wages conditional on history $h$ (identifying the four $\mu_{g,\gamma,h}$, given $\psi_{g,2}$); and (3) conditional covariance information $\sigma^2_{g,\gamma,h} = \E_{F_{g}}[(Y_{g,i,1} - \E[Y_{g,i,1}|h(i) = h]) Y_{g,i,2} | h(i)= h]$ for each history $h$ (identifying the four $\sigma^2_{g,\gamma,h}$, linearly given the mean parameters).
\end{remark}

\subsubsection{Supply-Side Effects in Financial Markets} \label{sec:example_banking_deregulation}

Our framework also informs analyses of policies reshaping local financial markets, such as U.S. interstate banking deregulation (e.g., \citealp{jayaratne1996finance}). This state-level policy, intended to alter market competition, allows studying effects on local credit supply (e.g., bank interest rates). A key challenge is that observed market outcomes (loan rates, quantities) are equilibrium objects reflecting both supply and demand. Controlling for firm-specific demand shocks is crucial to isolate supply-side policy effects, and our framework allows for this. 

Specifically, consider banks in state $g$ classified into three types relevant to deregulation: type 1 (e.g., small, incumbent in-state banks, often a benchmark), type 2 (e.g., large, incumbent in-state banks), and type 3 (e.g., newly entering or expanding out-of-state banks, whose presence is directly influenced by $W_g$). Let $Y_{g,i,b}$ be the outcome (e.g., interest rate) for firm $i$ from bank type $b$ in state $g$. Following seminal work by  \cite{khwaja2008tracing} we model these outcomes as:
\begin{equation*}
    Y_{g,i,b} = \gamma_{g,i} + \psi_{g,b} + \epsilon_{g,i,b}, \quad \mathbb{E}_{F_g}[\epsilon_{g,i,b}|K_{g,i}] = 0,
\end{equation*}
where $\gamma_{g,i}$ is a firm-specific effect capturing its idiosyncratic credit demand and overall creditworthiness (constant across bank types for firm $i$), and $\psi_{g,b}$ represents the average supply-side terms offered by bank type $b$ in state $g$. $K_{g,i}$ denotes the set of bank types with which firm $i$ interacts. Normalizing $\psi_{g,1} \equiv 0$, the parameters $\psi_{g,2}$ and $\psi_{g,3}$ represent the supply conditions of Type 2 and Type 3 banks relative to Type 1 banks. The state-level supply-side outcome vector is $\bm{\theta}_g \equiv (\psi_{g,2}, \psi_{g,3})'$. These parameters are identified from firms interacting with multiple bank types, as $\gamma_{g,i}$ is differenced out.\footnote{We assume that in each state there are firms that interact with multiple banks, in particular types 1 and 2, and types 1 and 3. This assumption is analogous to the existence of movers in standard AKM-type models.}

Suppose the goal is to quantify a causal effect of deregulation $W_g$ on the relative bank-type effects, postulating, for example, a linear relationship 
\begin{equation*}
    \psi_{g,b} = \alpha_{b} + \lambda_g + \beta W_{g,b},
\end{equation*}
where $\alpha_b$ are bank-type fixed effects, $\lambda_g$ are state-level fixed effects, and $W_{g,b}$  measures the exposure of bank type $b$ to the state-level policy $W_g$ (e.g., $W_g \mathbf{1}\{b=3\}$). A common empirical strategy then involves substituting this model for $\psi_{g,b}$ directly into the outcome equation for $Y_{g,i,b}$.\footnote{Historically, this strategy has not been used to asses the effect of the deregulation in the US due to the lack of access to credit registry data, but similar approaches have been used to investigate the effects of macro-prudential policies in Europe, e.g., \cite{jimenez2017macroprudential}.} The parameters, including $\beta$, are estimated by solving the following minimization problem:
\begin{equation*}
    \min_{\{\gamma_{g,i}\}, \{\alpha_b\}, \{\lambda_g\}, \beta} \sum_{g,i,b}(Y_{g,i,b} - \gamma_{g,i} - \mathbf{1} \{b \ne 1\}(\alpha_b + \lambda_g + \beta W_{g,b}))^2.
\end{equation*}
 This least squares estimator is an instance of a one-step GMM estimator \eqref{eq:gmm_est}.\footnote{Our abstract model relates $\bm{\theta}_g$ to $W_g$ via $\bm{\theta}_g = \bm{\alpha} + \Gamma \lambda_g + B W_g$ with $\Gamma=(1,1)^\top$ and $B=\text{diag}(\beta, \beta)$. Although $\Gamma^\top B \neq 0$ if $\beta \neq 0$, we can accommodate it by defining $B = \beta P_{\Gamma^\perp}$, where $P_{\Gamma^\perp} = \frac{1}{2}\begin{psmallmatrix} 1 & -1 \\ -1 & 1 \end{psmallmatrix}$ projects orthogonal to $\Gamma$. Effectively, $\beta$ is identified from the relationship between the policy difference ($W_{g,3}-W_{g,2}$) and the premium difference ($\psi_{g,3}-\psi_{g,2}$).}

\subsubsection{TFP and Trade Policy}
\label{sec:tfp_trade_example}

Our framework extends beyond linear regression to parameters identified within a structural model using general linear moment restrictions. Consider the estimation of firm productivity in industry $g$. Let $Y_{g,i,t}$ denote log output, and let $L_{g,i,t}$, $K_{g,i,t}$, and $M_{g,i,t}$ denote log inputs (labor, capital, and materials) for firm $i$. A standard Cobb-Douglas production function takes the form:
\begin{align*}
Y_{g,i,t} = \theta_{g}^l L_{g,i,t} + \theta_{g}^k K_{g,i,t} + \theta_{g}^m M_{g,i,t} + \omega_{g,i,t} + \eta_{g,i,t},
\end{align*}
where $\omega_{g,i,t}$ is unobserved firm-specific productivity and $\eta_{g,i,t}$ is an idiosyncratic error. Because firms choose variable inputs (labor and materials) based on their productivity $\omega_{g,i,t}$, standard OLS estimates of the elasticities are inconsistent. Identification requires solving a system of GMM moment conditions using lagged inputs as instruments (e.g., \cite{ackerberg2015identification}, see discussion below for details). Once the elasticities are identified, one can calculate Total Factor Productivity ($TFP_{g,i,t} := \omega_{g,i,t} + \eta_{g,i,t}$) as the residual.

Extensive literature suggests that access to imported intermediate inputs can boost firm productivity through learning effects, higher quality, and greater variety. For example, \cite{amiti2007trade} argue that reducing tariffs on intermediate inputs can have a significantly larger impact on productivity than reducing tariffs on final goods. In their setting (Indonesian manufacturing), the key mechanism is that cheaper imported inputs enable firms to access foreign technology embodied in those inputs. To be specific, let $\Delta FM_{g,i,t}$ indicate a change in the firm's import status and consider the following relationship:
\begin{align*}
\Delta TFP_{g,i,t} =  \delta_{g,t} + \tau_{g,t} \Delta FM_{g,i,t} + \epsilon_{g,i,t}.
\end{align*}
Here, $\tau_{g,t}$ represents the effect of entering global supply chains on productivity growth in industry $g$ at time $t$.

The causal link between importing and productivity ($\tau_{g,t}$) is a central object of interest. Specifically, researchers want to know if trade liberalization ($W_{g,t}$) effectively raises this ``import premium,'' thereby driving aggregate productivity growth. If $\tau_{g,t}$ were directly observed, the oracle analysis would estimate the effect of tariffs using a standard two-way fixed effects regression:
\begin{align*}
    \tau_{g,t} = \lambda_g + \alpha_t + \beta W_{g,t} + \nu_{g,t}.
\end{align*}
Assuming the standard parallel trends assumption holds, $\beta$ captures the causal effect of trade policy on the productivity gains from importing. 

Just as researchers stack DiD regressions, standard practice here often collapses the policy evaluation into a single step. For instance, one might regress estimated TFP growth on policy interactions without accounting for the first-stage estimation structure:
\begin{align*}
\Delta TFP_{g,i,t} =  \delta_{g,t} + (\lambda_g + \alpha_t + \beta W_{g,t}) \Delta FM_{g,i,t} + \varepsilon_{g,i,t},
\end{align*}
or an even more restrictive version of the same regression. The implicit weights in this regression depend on the variation of the micro-level treatment (here, the change in import status $\Delta FM_{g,i,t}$). If the trade policy $W_{g,t}$ affects the likelihood of importing—as is typically the case—it corrupts these weights. This correlation between the policy and the effective weights renders the one-step estimator inconsistent. A natural alternative is to follow the identification logic outlined above and directly relate the estimated $\hat{\tau}_{g,t}$ to policy variation.

We now provide additional details about the theoretical details. To simplify exposition, we focus on a single dynamic input, $L_{g,i,t}$ (the logic extends to multiple inputs with appropriate timing assumptions). The econometric specification for the production function is:
\begin{align*}
    Y_{g,i,t} &= \theta_{g}^l L_{g,i,t}+\omega_{g,i,t} + \eta_{g,i,t}, \\
    \omega_{g,i,t} &= \theta^{0}_{g} + \theta^{\omega}_g \omega_{g,i,t-1}  + \xi_{g,i,t}.
\end{align*}
The unobserved productivity component $\omega_{g,i,t}$ follows an AR(1) process. Key identifying assumptions for the $t \in \{1,2,3\}$ panel, particularly for estimation using data from $t=3$, include:
\begin{itemize}
    \item[(a)] The innovation to productivity $\xi_{g,i,3}$ is unpredictable by past inputs and output:
    \begin{align*}
        \mathbb{E}_{F_g}[\xi_{g,i,3}|L_{g,i,2},L_{g,i,1}, Y_{g,i,1}] = 0.
    \end{align*}
    \item[(b)] The idiosyncratic error $\eta_{g,i,t}$ (e.g., measurement error or unanticipated shock) is contemporaneously uncorrelated with inputs and unpredictable by past variables: 
    \begin{align*}
        \mathbb{E}_{F_g}[\eta_{g,i,3}|L_{g,i,2},L_{g,i,1}, Y_{g,i,1}] = 0, \quad  \mathbb{E}_{F_g}[\eta_{g,i,2}|L_{g,i,2}, L_{g,i,1}, Y_{g,i,1}] = 0.
    \end{align*}
\end{itemize}
Substituting $\omega_{g,i,t}$ from the AR(1) process into the production function and then expressing $\omega_{g,i,t-1}$ in terms of observables (i.e., $\omega_{g,i,2} = Y_{g,i,2} - \theta_g^l L_{g,i,2} - \eta_{g,i,2}$), we obtain an equation for $Y_{g,i,3}$:
\begin{equation*}
    Y_{g,i,3} = \theta^{0}_g(1-\theta^{\omega}_g) + \theta^{\omega}_g Y_{g,i,2} +  \theta_{g}^l L_{g,i,3} - \theta_{g}^l \theta^{\omega}_g L_{g,i,2} + \nu_{g,i},
\end{equation*}
where the composite error term is $\nu_{g,i} := \xi_{g,i,3} + \eta_{g,i,3} - \theta_{g}^{\omega} \eta_{g,i,2}$. Under assumptions (a) and (b), it follows that $\mathbb{E}_{F_g}[\nu_{g,i}|L_{g,i,2},L_{g,i,1},Y_{g,i,1}] =0$. 
Letting $\tilde{\theta}^{0}_g = \theta^{0}_g(1-\theta^{\omega}_g)$ and $\tilde{\theta}_{g}^l := -\theta_{g}^l \theta^{\omega}_g$, the equation becomes:
\begin{equation*}
    Y_{g,i,3} = \tilde{\theta}^{0}_g + \theta^{\omega}_g Y_{g,i,2} +  \theta_{g}^l L_{g,i,3} + \tilde{\theta}_{g}^l L_{g,i,2} + \nu_{g,i}, \quad \mathbb{E}_{F_g}[\nu_{g,i}|L_{g,i,2},L_{g,i,1},Y_{g,i,1}] =0.
\end{equation*}
Parameters $(\tilde{\theta}^{0}_g, \theta^{\omega}_g, \theta_{g}^l, \tilde{\theta}_{g}^l)$ are then identified using linear moment conditions with $L_{g,i,2}$, $L_{g,i,1}$, $Y_{g,i,1}$ (and a constant) as instruments. 

This identification of $(\theta_g^l, \theta_g^\omega, \theta_g^0)$ implies that TFP, defined as
\begin{align*}
    TFP_{g,i,t} =  Y_{g,i,t} - \theta_{g}^l L_{g,i,t},
\end{align*}
is identified. Key model-based parameters, $\delta_{g,t}$ and $\tau_{g,t}$, related to TFP growth, are then identified using moment conditions arising from the linear regression:
\begin{align*}
\Delta TFP_{g,i,t} =  \delta_{g,t} + \tau_{g,t} \Delta FM_{g,i,t} + \epsilon_{g,i,t}, \quad \mathbb{E}_{F_{g,t}}[\epsilon_{g,i,t} | \Delta FM_{g,i,t}] =0.
\end{align*}
The structural parameters of the production function ($\theta_g^l, \theta_g^\omega, \theta_g^0$) are separately identified from the coefficients in the GMM regression (e.g., $\theta_g^l$ is directly estimated, and $\theta_g^\omega$ allows recovery of $\theta_g^0$ from $\tilde{\theta}^{0}_g$).

\begin{remark}[Markups as Model-Based Outcomes]
The input elasticities estimated for TFP analysis are also central to firm markup estimation. Following insights from \cite{loecker2012markups}, we can link markups to a variable input's output elasticity and its observed expenditure share. This relationship defines average group-level markups as solutions to linear population moment conditions, thereby allowing researchers to use our framework to investigate the effects of policy on market competition.
\end{remark}

\clearpage

\section{Figures and tables}
    \label{app:figAndTableAppendix}
    \begin{figure}[ht]
    \centering
    \caption{Distribution of number of people with $x$ at $g$: $n_g(x)$}
    \label{fig:hist_ngx}
        \centering
        \includegraphics[width=0.95\textwidth]{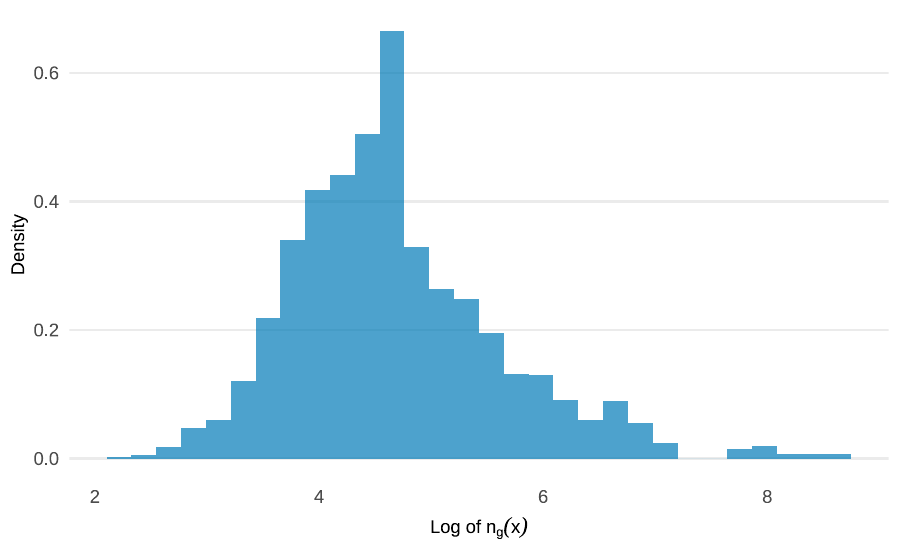}
    
    \caption*{
        \footnotesize \textit{Notes:}
        This figure presents the histogram of population with covariates $x$ (gender, birth year) who gave birth at municipality $g$. The x-axis shows the log of the population with $x$ and $g$ and y-axis shows its density. Due to the CBS confidentiality guidelines, $n_g(x) < 10$ are not included in the figure. The trimmed cells in the figure is 39 out of 4068 cells (0.96\%).
    }
    
\end{figure}

\begin{figure}[ht]
    \centering
    \caption{Child penalties by age at first childbirth}
    \label{fig:agg_CP}
    \begin{subfigure}[b]{0.9\textwidth}
        \centering 
        \caption{Earnings}
        \label{fig:agg_earn}
            \includegraphics[width=.95\linewidth, keepaspectratio=true]{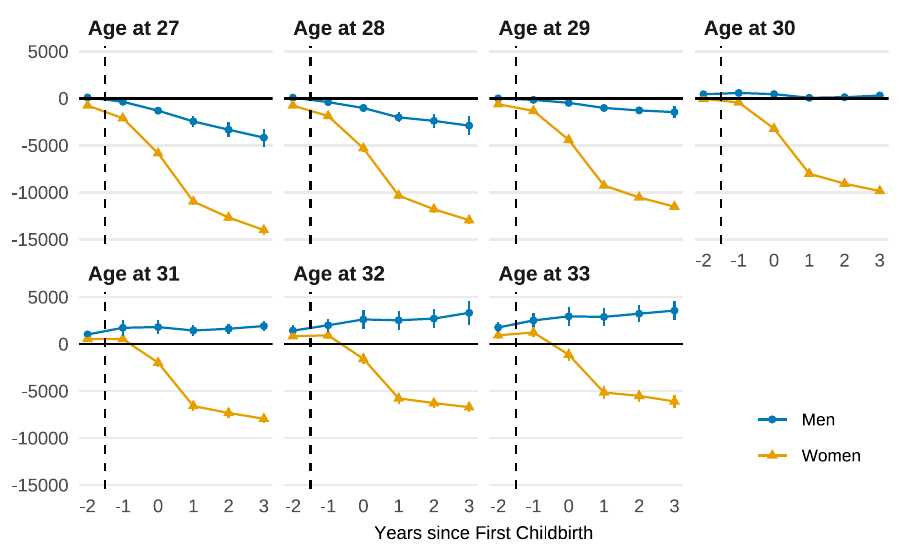} 
    \end{subfigure}

    \begin{subfigure}[b]{0.9\textwidth}
        \centering
        \caption{Participation}
        \label{fig:agg_particip}
            \includegraphics[width=.95\linewidth, keepaspectratio=true]{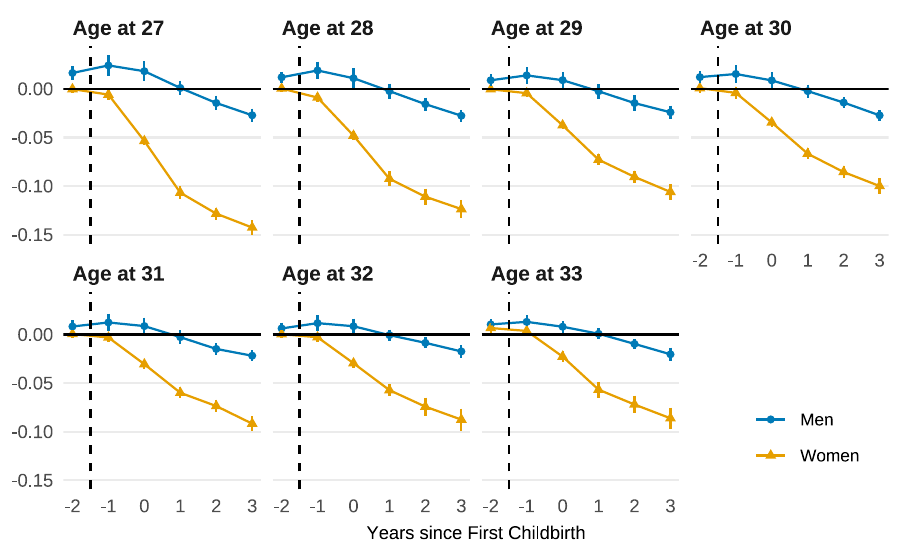} 
    \end{subfigure}
    \caption*{
        \footnotesize \textit{Notes:} This figure presents CP estimates for yearly earnings, aggregated by birth cohort, gender, municipality of residence at time of pregnancy, and the age of first childbirth, as described in Section \ref{sec:measurement}. Figure \ref{fig:agg_earn} reports the average CP for total yearly earnings. Similarly, Figure \ref{fig:agg_particip} reports the CP for the labor-market participation margin.
    }
    
\end{figure}

\begin{table}[ht]
    \centering
    \caption{Rich MD specification -- the childcare provision expansion on CP (fathers)}
    \label{tab:reg_flex_m}
    \begin{minipage}{.7\linewidth}
        \subcaption{Earnings}
        \label{tab:reg_flex_earn_m}
        \resizebox{\columnwidth}{!}{%
            \resizebox{\ifdim\width>\linewidth 1\linewidth\else\width\fi}{!}{
\begin{talltblr}[         %% tabularray outer open
entry=none,label=none,
note{}={+ p < 0.1, * p < 0.05, ** p < 0.01},
]                     %% tabularray outer close
{                     %% tabularray inner open
colspec={Q[]Q[]Q[]Q[]Q[]Q[]Q[]},
column{2,3,4,5,6,7}={}{halign=c,},
column{1}={}{halign=l,},
hline{14}={1,2,3,4,5,6,7}{solid, black, 0.05em},
}                     %% tabularray inner close
\toprule
& $\hat{\tau}_{g, -2}$ & $\hat{\tau}_{g, -1}$ & $\hat{\tau}_{g, 0}$ & $\hat{\tau}_{g, 1}$ & $\hat{\tau}_{g, 2}$ & $\hat{\tau}_{g, 3}$ \\ \midrule %% TinyTableHeader
$CCI_{g, b+e-2}$ & 590.5 & -148.5 & 189.5 & 674.9 & 1026.0 & 1494.5 \\
& (397.8) & (804.4) & (1113.7) & (1116.5) & (1111.9) & (1349.7) \\
$CCI_{g, b+e-1}$ &  & 1704.8** & 1532.7+ & 1550.1 & 940.2 & 638.2 \\
&  & (592.7) & (804.7) & (989.2) & (923.6) & (985.9) \\
$CCI_{g, b+e}$ &  &  & 1153.8* & 575.2 & 1394.8+ & 780.2 \\
&  &  & (560.5) & (666.6) & (709.5) & (840.9) \\
$CCI_{g, b+e+1}$ &  &  &  & 1209.2+ & 859.0 & 1976.2* \\
&  &  &  & (720.0) & (718.8) & (880.4) \\
$CCI_{g, b+e+2}$ &  &  &  &  & 1356.1* & 986.1 \\
&  &  &  &  & (675.0) & (981.8) \\
$CCI_{g, b+e+3}$ &  &  &  &  &  & 208.0 \\
&  &  &  &  &  & (850.9) \\
N & 11,739 & 11,739 & 11,739 & 11,739 & 11,739 & 11,739 \\
$R^2$ & 0.126 & 0.189 & 0.205 & 0.219 & 0.227 & 0.251 \\
FE: Municipality $g$ & X & X & X & X & X & X \\
FE: $B_{g, i} \times E_{g, i}$ & X & X & X & X & X & X \\
FE: $(B_{g, i} \times E_{g, i})S_g$ & X & X & X & X & X & X \\
\bottomrule
\end{talltblr}
}

        }
    \end{minipage}

    \begin{minipage}{.7\linewidth}
        \subcaption{Participation}
        \label{tab:reg_flex_particip_m}
        \resizebox{\columnwidth}{!}{%
            \resizebox{\ifdim\width>\linewidth 1\linewidth\else\width\fi}{!}{
\begin{talltblr}[         %% tabularray outer open
entry=none,label=none,
note{}={+ p < 0.1, * p < 0.05, ** p < 0.01},
]                     %% tabularray outer close
{                     %% tabularray inner open
colspec={Q[]Q[]Q[]Q[]Q[]Q[]Q[]},
column{2,3,4,5,6,7}={}{halign=c,},
column{1}={}{halign=l,},
hline{14}={1,2,3,4,5,6,7}{solid, black, 0.05em},
}                     %% tabularray inner close
\toprule
& $\hat{\tau}_{g, -2}$ & $\hat{\tau}_{g, -1}$ & $\hat{\tau}_{g, 0}$ & $\hat{\tau}_{g, 1}$ & $\hat{\tau}_{g, 2}$ & $\hat{\tau}_{g, 3}$ \\ \midrule %% TinyTableHeader
$CCI_{g, b+e-2}$ & -0.010 & -0.008 & -0.018 & -0.010 & -0.022 & -0.007 \\
& (0.010) & (0.014) & (0.014) & (0.015) & (0.015) & (0.016) \\
$CCI_{g, b+e-1}$ &  & 0.000 & 0.012 & 0.008 & 0.013 & -0.009 \\
&  & (0.014) & (0.014) & (0.017) & (0.015) & (0.015) \\
$CCI_{g, b+e}$ &  &  & -0.025 & -0.004 & 0.002 & 0.007 \\
&  &  & (0.018) & (0.019) & (0.018) & (0.014) \\
$CCI_{g, b+e+1}$ &  &  &  & -0.015 & -0.030* & -0.014 \\
&  &  &  & (0.009) & (0.012) & (0.011) \\
$CCI_{g, b+e+2}$ &  &  &  &  & 0.010 & 0.003 \\
&  &  &  &  & (0.012) & (0.012) \\
$CCI_{g, b+e+3}$ &  &  &  &  &  & -0.003 \\
&  &  &  &  &  & (0.011) \\
N & 11,739 & 11,739 & 11,739 & 11,739 & 11,739 & 11,739 \\
$R^2$ & 0.070 & 0.107 & 0.095 & 0.076 & 0.076 & 0.083 \\
FE: Municipality $g$ & X & X & X & X & X & X \\
FE: $B_{g, i} \times E_{g, i}$ & X & X & X & X & X & X \\
FE: $(B_{g, i} \times E_{g, i})S_g$ & X & X & X & X & X & X \\
\bottomrule
\end{talltblr}
}

        }
    \end{minipage}

    \vspace{0.5em}
    \caption*{
        \footnotesize \textit{Notes:} These tables present the effect of the childcare provision expansion on child penalties (CP) of fathers in earnings (Panel (a)) and labor force participation (Panel (b)). See Section \ref{sec:rich_md} for more details and Equation~\eqref{eq:rich_md} for the specification. Parenthesis shows the clustered standard errors by municipality $g$.
    }
\end{table}

\clearpage

\section{Data}
    \label{sec:data}

\subsubsection*{Data Sources}
\label{sec:dataSources}

We use administrative data from the Central Bureau of Statistics Netherlands (CBS) on the universe of Dutch residents. Different data sources, such as municipality registers or tax records, are matched through unique individual or household anonymized identifiers. The following section presents the main variables used and sample construction. 

\paragraph{Tax and Employment Records}
Our primary data source is an extensive annual-level employer-employee data set derived from tax records (\textit{baansommentab}) covering 1999 to 2016. We analyze two labor market outcomes: unconditional earnings and employment. Employment is specified as having a job based on an employment contract between a firm and a person, excluding self-employment. Second, earnings data consist of yearly gross earnings after social security contributions but before taxes and health insurance contributions from official tax data.

\paragraph{Demographic and Education Information}
To enrich our understanding of the workforce, we incorporate demographic data into our analysis (\textit{gbapersoontab}). This includes birth year, date of death, sex, and annual information on the municipality of residence, household composition, marital status, and migration spells (\textit{gbaadresobjectbus} and \textit{vslgwbtab}). A unique aspect of our demographic data is the inclusion of a parent-child key (\textit{kindoudertab}). We use information on birthdates and the linkage between parents and their children to determine the first child for all legal parents, which may include both adoptive and biological parents. 
%Lastly, we also observe the educational attainment at each point in time (\textit{hoogsteopltab}) and use the highest level of education attained by 2022, which we classify into three levels: high school, vocational training, and bachelor's degree. We exclude individuals with higher education (MA and PhD) and lower education (below high-school) subpopulations for two reasons: (a) they form a smaller share of the population, and (b) fertility and labor market decisions will likely follow a different pattern in those groups.\footnote{These different life-cycle earnings patterns translate to violations of our statistical model, requiring another sample restriction.} 

\paragraph{Childcare Provision Data}
An integral part of our study involves examining the role of childcare in labor market participation. To this end, we use records on childcare service providers using the firm's job classification (\textit{betab}), and data on job location that we use to compute our index of childcare supply per municipality (given from \textit{gemstplaatsbus/gemtplbus/ngemstplbus}). The job location data set contains each worker's municipality and firm ID, which we merge with the firm classification data.

\subsubsection*{Sample Definition}
\label{sec:sampDef}

A key aspect of our study is the examination of labor market outcomes around the time of first childbirth. We restrict the sample to individuals born in 1993 or earlier to ensure we observe labor market outcomes at sufficiently mature ages. To capture transitions into parenthood, we include only those whose age at first birth was below 44, as observed from 2003 onward. To ensure adequate labor market attachment before parenthood, we further restrict the sample to individuals who became parents at least 6 years after the typical graduation age for their highest educational attainment: 24 years for high school graduates, 26 years for vocational degree holders, and 27 years for those with a bachelor’s degree. This approach provides a balanced panel of pre-parenthood labor market trajectories while minimizing censoring concerns.

\clearpage

\section{Dutch Childcare Act Reform}
    \label{sec:childcareAct}
    
\subsection{Key Features of the Reform}
The \textit{Dutch Childcare Act of 2005} (\textit{Wet kinderopvang}) radically reformed childcare funding, provision, and regulation. It replaced a patchwork of local subsidies with a national, demand-driven system. The core expansion mechanisms included:

\paragraph{Demand-Side Funding.} The Act shifted from supply-side grants to demand-led financing, giving parents direct subsidies or tax credits for childcare. Costs are now shared by parents, employers, and the state, with compulsory employer contributions \citep{ChildcareCanada2013}. This tripartite funding greatly reduced out-of-pocket fees for families—on average cutting the effective parental fee by half over 2005--2009 \citep{Bettendorf2015}. All forms of formal care (daycare centers, childminders, after-school programs) became eligible for the same central subsidy, ending the old differentiation between “subsidized” and “unsubsidized” spots. Even informal care by relatives or licensed in-home providers (so-called \textit{gastouder} or “guestparent” care) was brought under the subsidy scheme \cite{ChildcareCanada2013}.

\paragraph{Market Liberalization.} By decoupling funding from public provision, the Act introduced market forces into childcare. For-profit providers were explicitly allowed to enter and expand, competing on an equal footing with non-profits in a newly privatized market \cite{NoaillyVisser2009}. The expectation was that increased competition and profit incentives would spur new capacity. Indeed, following the reform, childcare supply grew fastest in high-demand areas and was led largely by for-profit firms, while non-profit providers saw their market share decline (especially in less affluent regions). This outcome reflected providers gravitating toward municipalities with greater demand and purchasing power, raising some concerns about equitable access in low-income or rural areas. Overall, however, the liberalization unlocked rapid growth in the number of childcare facilities nationwide \cite{NoaillyVisser2009}.

\paragraph{Regulatory Consistency.} The Act sought to harmonize and simplify regulation of childcare. It introduced a single national quality framework (with light-touch regulation) applying to all providers \cite{ChildcareCanada2013}. This replaced the prior mix of local rules and employer-based arrangements, ensuring consistent standards (e.g., safety, staff qualifications, and child-to-staff ratios) across the country. Notably, legally mandated child-to-staff ratios (set in 1996) remained in place to safeguard minimum quality \cite{Decree1996}, but administrative burdens were kept low to encourage new entrants. In essence, the reform fully privatized childcare provision but under a uniform set of basic rules. Subsequent legislation in 2010 further cemented this consistency with a single statutory quality code for all early-childhood education and care services.

\subsection{Implementation Timeline and Childcare Expansion}
The Childcare Act took effect on \textbf{1 January 2005}, marking the start of a rapid expansion in Dutch childcare usage and supply. Key stages in its rollout include:

\begin{itemize}
  \item \textbf{2005 -- Transition Year:} The immediate impact was modest. The new funding system unified subsidies, slightly raising subsidies for previously “unsubsidized” parents and reducing them for the highest-income group, largely balancing out. Public spending on childcare actually dipped slightly in 2005, and the growth in childcare slots did not yet accelerate. Thus, no major change in labor supply was expected in this first year of reform implementation. The groundwork was laid, however, for broader participation: all formal childcare now qualified for support, and awareness of the new scheme grew among parents and providers.
  \item \textbf{2006--2007 -- Surge in Generosity and Supply:} In the two years after 2005, the government dramatically increased the childcare subsidy rates. By 2007, the average parental co-payment share dropped from about 37\% of the true cost to just 18\%. Middle-income families saw subsidy rates rise by 20--40 percentage-points, and even higher-income families gained substantially. These changes effectively halved the cost of childcare for families within a short period. As a result, demand responded sharply: enrollment in formal childcare surged, and providers raced to expand capacity. The number of childcare places grew rapidly after 2006, with especially fast growth in daycare and out-of-school care participation. To support school-aged children’s care, a 2007 mandate required all primary schools to ensure out-of-school care was available (often by coordinating with childcare organizations).
  \item \textbf{Post-2007 -- Consolidation and Further Growth:} Generous funding continued through the late 2000s. Public expenditure on childcare subsidies climbed from about €1.0 billion in 2004 to €3.0 billion by 2009, roughly 0.5\% of GDP. The take-up of childcare subsidies expanded across all income groups (with low-income families eligible for nearly free care). The increased demand was met predominantly by private centers and childminders entering the market under the light regulatory regime. By the end of 2009, the system had matured into a full-fledged childcare market with substantially higher coverage than pre-reform. The joint introduction of an expanded in-work tax credit for parents (the \textit{combinatiekorting}) during 2005--2009 also provided additional incentive for parental employment \cite{Bettendorf2015}.
\end{itemize}

\end{document}